\documentclass[11pt]{article}
\usepackage[T1]{fontenc}
\usepackage{amsfonts}
\usepackage{amsmath}
\usepackage{enumerate}
\usepackage{amssymb}
\usepackage{amsthm}
\usepackage{bbm}
\usepackage{bm}
\usepackage{mathrsfs}
\usepackage{verbatim}
\usepackage{setspace}
\usepackage{color}
\usepackage{pdfsync}
\usepackage{nicefrac}
\theoremstyle{plain}
\newtheorem{theorem}{Theorem}[section]
\newtheorem{proposition}[theorem]{Proposition}
\newtheorem{lemma}[theorem]{Lemma}
\newtheorem{corollary}[theorem]{Corollary}

\theoremstyle{definition}

\newtheorem{remark}[theorem]{Remark}

\newtheorem{assumption}[theorem]{Assumption}

\theoremstyle{remark}

\renewenvironment{thebibliography}[1]{%
\begin{oldthebibliography}{#1}%
\setlength{\baselineskip}{.9em}
\linespread{1}
\small
\setlength{\parskip}{0.3ex}%
\setlength{\itemsep}{.5em}%
}%
{%
\end{oldthebibliography}%
}

\newcommand{\E}{\mathbb{E}}
\newcommand{\F}{\mathbb{F}}

\newcommand{\N}{\mathbb{N}}
\renewcommand{\P}{\mathbb{P}}
\newcommand{\Q}{\mathbb{Q}}
\newcommand{\R}{\mathbb{R}}
\renewcommand{\S}{\mathbb{S}}

\newcommand{\cC}{\mathcal{C}}
\newcommand{\cD}{\mathcal{D}}
\newcommand{\cE}{\mathcal{E}}
\newcommand{\cF}{\mathcal{F}}
\newcommand{\cG}{\mathcal{G}}
\newcommand{\cH}{\mathcal{H}}

\newcommand{\cK}{\mathcal{K}}

\newcommand{\cM}{\mathcal{M}}
\newcommand{\cN}{\mathcal{N}}

\newcommand{\cP}{\mathcal{P}}

\newcommand{\cS}{\mathcal{S}}

\newcommand{\cX}{\mathcal{X}}

\newcommand{\fP}{\mathfrak{P}}

\newcommand{\fM}{\mathfrak{M}}

\DeclareMathOperator{\proj}{proj}

\newcommand{\ov}{\overline}
\numberwithin{equation}{section}

\usepackage[pdfborder={0 0 0}]{hyperref}
\hypersetup{
  urlcolor = black,
  pdfauthor = {Daniel Bartl, Michael Kupper, Ariel Neufeld},
  pdfkeywords = {T.B.D.;
},
}

\begin{document}

\title{ \vspace{-2.4em} 
Duality Theory for Robust Utility Maximisation 
\date{\today}
\author{
  Daniel Bartl%
  \thanks{
  Department of Mathematics, University of Vienna, \texttt{daniel.bartl@univie.ac.at}.
  Financial support by the Austrian Science Fund 
  (FWF) under project P28661
  and by the Vienna Science and Technology Fund (WWTF) under project MA16-021
  is gratefully acknowledged.
  }
  \and
    Michael Kupper%
     \thanks{
     Department of Mathematics and Statistics, University of Konstanz, \texttt{kupper@uni-konstanz.de}.
     }
  \and
  Ariel Neufeld%
   \thanks{
   Division of Mathematical Sciences, NTU Singapore, \texttt{ariel.neufeld@ntu.edu.sg}.
   Financial support by his 
   Nanyang Assistant Professorship Grant (NAP Grant) \emph{Machine Learning based Algorithms in Finance and Insurance}
   is gratefully acknowledged.
   }
 }
}
\maketitle \vspace{-1.2em}

\begin{abstract}
In this paper we present a duality theory for the robust utility maximisation problem in continuous time for utility functions defined on the positive real axis.
Our results are inspired by -- and can be seen as the robust analogues of -- the seminal work of Kramkov~\&~Schachermayer~\cite{KramkovSchachermayer.99}.
Namely, we show that if the set of attainable trading outcomes and the set of pricing measures satisfy a bipolar relation, then the utility maximisation problem is in duality with a conjugate problem.
We further discuss the existence of optimal trading strategies.
In particular, our general results include the case of logarithmic and power utility, and they apply to drift and volatility uncertainty.
%
\end{abstract}

\vspace{.9em}

{\small

}

\section{Introduction}
The goal of this paper is to develop a duality theory for the robust utility maximisation problem. 
Given a utility function $U\colon(0,\infty)
\to \R$ which is nondecreasing and concave,  one defines the robust utility maximisation problem as
\begin{equation}\label{primal-classical}
u(x):= \sup_{g \in \mathcal{C}(x)}\inf_{\P \in \cP} \E_\P[U(g)].
\end{equation}
Here $\cP$ denotes a set of possibly nondominated probability  measures, and $\mathcal{C}(x)\equiv x\cC$ is a set of random variables. Financially speaking, the set $\cP$ represents the set of possible candidates for the real-world measure, which is not known to the portfolio manager who tries to solve the maximisation problem, whereas the set $\cC(x)$ represents all the possible portfolio values at terminal time $T$ available with initial capital $x>0$. Note that if $\cP=\{\P\}$ is a singleton, then the robust utility maximisation problem coincides with the \textit{classical utility maximisation problem} and has the financial interpretation  that the portfolio manager (believes to) know the real-world measure $\P$. 
The robust utility maximisation problem with respect to nondominated probability measures $\cP$ has been already widely studied; see
\cite{bartl2019exponential, bartl2019limmed, blanchard2018multiple, NeufeldSikic2018robust, neufeld2019nonconcave,rasonyi2018utility}
for works in a discrete-time setting, and
\cite{biagini2017robust, chau2019robust,DenisKervarec.13,FouquePunWong.16, guo2019robust, ismail2019robust, liang2018robust, LinRiedel.14, lin2020horizon, MatoussiPossamaiZhou.12utility, neufeld2018robust, NeufeldSikic2018robust, pham2018portfolio, pun2018g, TevzadzeToronjadzeUzunashvili.13,uugurlu2020robust, yang2019constrained}
for works in a continuous-time setting. 
To the best of our knowledge, the only paper so far which provides a duality theory for the robust utility maximisation problem with nondominated probability measures in a continuous-time setting is the one by Denis~\&~Kervarec \cite{DenisKervarec.13}.  Indeed, Denis~\&~Kervarec \cite{DenisKervarec.13} provide a duality theory for the robust utility maximisation problem under drift and volatility uncertainty, but under the strong assumptions that the utility function is bounded and that the set of trading strategies possess some continuity (as a functional of the stock price). In addition, the corresponding volatility matrix is required to be of diagonal form.

In order to get an idea how one could try to identify the dual problem 
for the robust utility maximisation problem, let us recall the main idea in the seminal paper by Kramkov~\&~Schachermayer~\cite{KramkovSchachermayer.99} which provides a complete duality theory in the classical case where $\cP=\{\P\}$ is a singleton.
Consider the conjugate function  $V\colon(0,\infty)\to \R$ defined by 
\begin{equation*}
V(y):=\sup_{x\geq 0} \big[U(x)-xy\big], \quad y>0.
\end{equation*}
In the classical case where $\cP=\{{\P}\}$ is a singleton, Brannath~\&~Schachermayer~\cite{brannath1999bipolar} developed a bipolar theorem on $L^0_+(\P)$, despite the fact that $L^0_+(\P)$ is generally not a locally convex space, by considering the dual pairing $(L^0_+({\P}),L^0_+({\P}))$ endowed with the scalar product $\langle g, h \rangle:=\E_{{\P}}[gh]$ (see also {\v Z}itkovi\'c~\cite{Zitkovic.02} for a conditional version). This bipolar theorem then allows to identify the dual optimisation problem and to prove that the corresponding optimisation problems are conjugates. More precisely, let $\cC\subseteq L^0_+(\P)$ and define the polar set
\begin{equation}\label{intro:D-KS}
\cD:=\big\{h\in L^0_+(\P)\colon \E_{{\P}}[gh]\leq 1 \mbox{ for all } g \in \cC \big\}.
\end{equation}
Then one can define the dual optimisation problem by
\begin{equation}\label{intro:dual-classical}
v(y):=\inf_{h \in \cD}\E_{\P}[V(yh)].
\end{equation}
By the bipolar relation that 
\begin{align}
\cD&= \big\{h \in L^0_+(\P)\colon \E_{\P}[gh]\leq 1 \mbox{ for all } g \in \cC\big\}\label{intro:bipolar-1}\\
\cC&= \big\{g \in L^0_+(\P)\colon \E_{\P}[gh]\leq 1 \mbox{ for all } h \in \cD\big\} \label{intro:bipolar-2}
\end{align}
together with a minimax argument, Kramkov~\&~Schachermayer~\cite[Theorem~3.1]{KramkovSchachermayer.99} proved that indeed $u$ and $v$ are conjugates, namely that
\begin{equation}
\label{intro:conjugate}
\begin{split}
u(x)&=\inf_{y\geq 0}\big[v(y)+xy\big], \quad x>0,\\
v(y)&=\sup_{x\geq 0}\big[u(x)-xy\big], \quad y>0.
\end{split}
\end{equation}

However, in the robust analogue where $\cP$ is not a singleton, it is not clear how to find a suitable dual pairing $(\cX,\cX^*)$ for $\mathcal{X}\supseteq \cC$  such that a bipolar relation \eqref{intro:bipolar-1}~\&~\eqref{intro:bipolar-2} holds.
Our approach is the following. Instead of working on an arbitrary measurable space $(\Omega,\cF)$, we impose that $\Omega$ is a Polish space endowed with its Borel $\sigma$-field. This allows us to use the natural dual pairing $(C_b(\Omega),\fP(\Omega))$ consisting of the bounded continuous functions $C_b\equiv C_b(\Omega)$ together with the set of Borel probability measures $\fP(\Omega)$ on $\Omega$. Given a set $\cC$ of nonnegative measurable functions defined on $\Omega$, we then define its polar set by
\begin{equation}\label{intro:D-our}
\cD:=\big\{ \Q\in \fP(\Omega)\colon \E_{\Q}[g]\leq 1 \mbox{ for all } g \in \cC\big\}.
\end{equation}
This allows us to formulate a bipolar relation on the subset $C_b$, namely we require that\footnote{We denote by $C_b^+:=\{g \in C_b(\Omega)\colon g\geq 0\}$.}
\begin{align}
\cD&= \big\{\Q \in \fP(\Omega) \colon \E_{\Q}[g]\leq 1 \mbox{ for all } g \in \cC\cap C_b\big\}\label{intro:bipolar-our-1}\\
\cC\cap C_b&= \big\{g \in C_b^+\colon \E_{\Q}[g]\leq 1 \mbox{ for all } \Q \in \cD\big\} \label{intro:bipolar-our-2}. 
\end{align}
In our first result, we show that if $\cC$ is a set of nonnegative measurable functions and $\cD$ is a set of probability measures defined by  \eqref{intro:D-our} such that the bipolar relation \eqref{intro:bipolar-our-1}~\&~\eqref{intro:bipolar-our-2} holds, then the functions $u$ and $v$ defined by \eqref{primal-classical} and\footnote{We set
	$V(-\infty):=\infty$, and $\tfrac{d\Q}{d\P}:=-\infty$ if $\Q$ is not absolutely continuous with respect to $\P$.}
\begin{equation*}
v(y):=\inf_{\Q \in \cD}\inf_{\P \in \cP}\E_{\Q}\big[V(y\tfrac{d\Q}{d\P})\big]
\end{equation*}
 indeed satisfy the conjugate relation \eqref{intro:conjugate}; see Theorem~\ref{thm:main}.

At first glance, 
the bipolar relation \eqref{intro:bipolar-our-1}~\&~\eqref{intro:bipolar-our-2} 
might seem restrictive. However, it turns out that in fact, this bipolar relation is \textit{naturally} satisfied in the context of drift and volatility uncertainty. More precisely, let $\Omega=C([0,T],\R^d)$ and consider the set of probability measures  $\cP:=\cP^{ac}_{sem}(\Theta)$ for which the canonical process $(S_t)_{0\leq t \leq T}$ is a semimartingale with differential characteristics taking values in $\Theta$. Then, we define 
\begin{equation}\label{intro:C-our}
\mathcal{C}:=\big\{g \colon \exists\, H \in \cH \mbox{ such that } g\leq 1 + (H\cdot S)_T \ \P\mbox{-a.s. } \forall\, \P \in \cP\big\},
\end{equation}
namely $\mathcal{C}$ is the set of $\cP$-quasi surely superreplicable claims, and $\cD$ is its polar set defined in \eqref{intro:D-our}. As admissibility condition on the set of hedging strategies $\cH$ we require for each $H \in \cH$ that the stochastic integral satisfies $(H\cdot S)\geq -c$ 
$\P$-a.s.\ $\forall\, \P\in \cP$ 
for some constant $c>0$ 
(where $c$ can depend on both $H$ and $\mathbb{P}$),
similar to the classical admissibility condition.
This setting can be seen as the robust analogue to the setting of Kramkov~\&~Schachermayer~\cite[Theorem~2.1]{KramkovSchachermayer.99}. We show in (the proof of) Theorem~\ref{thm:main-app}
that for $\cP:=\cP^{ac}_{sem}(\Theta)$ together with $\cC$ and $\cD$ defined in \eqref{intro:C-our} and \eqref{intro:D-our}, the bipolar relation \eqref{intro:bipolar-our-1}~\&~\eqref{intro:bipolar-our-2} is satisfied \textit{automatically}.

Let us explain why the bipolar relation \eqref{intro:bipolar-our-1}~\&~\eqref{intro:bipolar-our-2} holds.
 To see that \eqref{intro:bipolar-our-1} holds, note  that since $\cC$ is the set of superhedgeable claims, every element in $\cD$ satisfies that $\E_{\Q}[(H\cdot S)_T]\leq 0$ for every $H \in \cH$. In other words, $\cD$ can be seen as the set of separating measures (in the notion of Kabanov~\cite{kabanov1997ftap}). Moreover, since $S$ has continuous sample paths, it is well-known that the set of separating measures coincides with the set of local martingale measures; see, e.g., \cite[Lemma~5.1.3, p.74]{delbaen2006mathematics}). Having this in mind, the condition \eqref{intro:bipolar-our-1} means that the set of separating measures is already determined by the superhedgeable claims which are continuous. We show that this is indeed true due to the Polish structure of $\Omega$; see Proposition~\ref{prop:separation-app} and Proposition~\ref{prop:D-C-Cb-polar}. 

To see that \eqref{intro:bipolar-our-2} holds, we need to show that every nonnegative and continuous bounded claim $g$ can be superhedged. In classical theory where $\cP=\{\P\}$ is a singleton, this has been proved  (even for measurable claims)  by El~Karoui~\&~Quenez \cite{el1995dynamic} and Kramkov~\cite{kramkov1996optional} in the following way. First they show that there exists a nonnegative process $(Y_t)_{0\leq t \leq T}$ with $Y_T=g$ such that 
$Y$ is a $\Q$-supermartingale for every equivalent local martingale measure $\Q$. Then, the optional decomposition theorem guarantees that $g=Y_T\leq 1+ (H\cdot S)_T$ for some hedging strategy $H$. The robust analogue of the results in  El~Karoui~\&~Quenez \cite{el1995dynamic} and Kramkov~\cite{kramkov1996optional} has been recently developed by Soner~\&~Nutz~\cite{nutz2012superhedging}, Neufeld~\&~Nutz~\cite{NeufeldNutz.12}, and Nutz~\cite{nutz2015robust} which is compatible with the set $\cP:=\cP^{ac}_{sem}(\Theta)$. This allows us to prove that indeed \eqref{intro:bipolar-our-2} holds; see Proposition~\ref{prop:C-C-b-char}.

Like in the classical result by Kramkov~\&~Schachermayer~\cite[Theorem~3.1]{KramkovSchachermayer.99}, we need to apply a minimax argument (together with the bipolar relation) in order to prove Theorem~\ref{thm:main}. To that end the lower-semicontinuity of the map
\begin{equation*}
\P\mapsto \E_{\P}[U(g)]
\end{equation*}
for every $g \in C_b^+$ is crucial. For this reason, we need to impose that $U$ is bounded from below, i.e. $U(0):=\lim_{x \downarrow 0} U(x)>-\infty$. To obtain the corresponding result also for utility functions which are unbounded from below (i.e.\ $U(0)=-\infty$), we observe
that the utility functions $U_n\colon(0,\infty)\to \R$, $n\in \N$, defined by $U_n(x):=U(x+\tfrac{1}{n})$, $x>0$, are bounded from below. Hence we apply Theorem~\ref{thm:main} to each $U_n$ together with a limit argument to prove that the corresponding duality result  also holds true for utility functions which are unbounded from below; see Theorem~\ref{thm:main-2}. 
However, since the set of probability measures $\cP$ is typically non-dominated, we cannot apply tools like the Koml\'os theorem to obtain a limit. This is why for Theorem~\ref{thm:main-2} we impose that medial limits exist (see also Theorem~\ref{thm:main-app-2}).

From a technical point of view, next to the common assumptions that $\Theta$ is convex and compact, we additionally assume that $\Theta$ satisfies a uniform ellipticity condition; see Assumption~\ref{ass:invertible}. This technical assumption is crucial to
	connect $\cP=\cP^{ac}_{sem}(\Theta)$ with a corresponding set of martingale measures, however, imposes  that each measure $\P \in \cP$ corresponds to a complete financial market. We refer to Remark~\ref{rem:Theta} for a more detailed discussion regarding the importance of this uniform ellipticity condition.
	
Finally, due to the generality of the utility function $U$, we impose some conditions to make sure that $u(x)<\infty$ for every $x>0$,\footnote{This condition is standard in the literature dealing with utility maximisation.}
 together with some integrability conditions which allow us to conclude Theorem~\ref{thm:main-2} from Theorem~\ref{thm:main} by a limit argument; see  Assumptions~\ref{ass:u-finite}~\&~\ref{ass:u-bar-finite} and Assumptions~\ref{ass:U-I}~\&~\ref{ass:v-finite}. We point out that these conditions are naturally satisfied by all the relevant utility functions. Indeed, these conditions are automatically satisfied for utility functions which are bounded from above, as well as, in the setting of simultaneously drift and volatility uncertainty described above\footnote{i.e.\ when $\cP:=\cP^{ac}_{sem}(\Theta)$ and $\cC$ and $\cD$ are defined by \eqref{intro:C-our} and \eqref{intro:D-our}}, also for the common utility functions like  logarithmic, power, and exponential utilities; see Corollary~\ref{co:main-app}~\&~\ref{co:main-app-2} and Remark~\ref{rem:ass:U-I}~\&~\ref{rem:ass:U-I-2}.

 The remainder of this paper is organised as follows. 
 In Section~\ref{sec:abstract}, we present the main results in the abstract setting.
 In Section~\ref{sec:application}, we state our main results in the setting of drift and volatility uncertainty.
The proofs of the results of Section~\ref{sec:abstract} are provided in Section~\ref{sec:proof}, whereas the proofs of the  results of Section~\ref{sec:application} are provided in Section~\ref{sec:proof-app}.

%
\section{Main results in the abstract setting}\label{sec:abstract}
We fix a time horizon $T\in (0,\infty)$ and a Polish space $\Omega$ endowed with its Borel $\sigma$-field $\mathcal{F}$. We denote by $\fP(\Omega)$ the set of all Borel probability measures endowed with the topology induced by the weak convergence, becoming itself a Polish space. Without further mentioning, we interpret $\fP(\Omega)$ as a subset of the set of all nonnegative, finite measures. In addition, we denote by $\mathcal{F}^*:=\cap_{\P\in \fP(\Omega)} \mathcal{F}^\P$ the universal $\sigma$-field. Moreover, we denote by $C_b$ the set of all continuous functions from $\Omega$ to $\R$ and $C_b^+$ denotes the subset of all nonnegative functions in $C_b$. Furthermore, we denote by $\S^d_+$ the set of all symmetric positive semi-definite matrices in $\R^{d\times d}$ and by 
$\S^d_{++}\subseteq \S^d_+$ the set of all positive definite matrices in $\S^d_+$. In addition, we say that $A\leq B$ holds for $A,B \in \R^{d\times d}$ if $B-A$ is positive semi-definite. 
%
%

We first fix a nonempty set of probability measures $\cP\subseteq \fP(\Omega)$, which represents the set of possible candidates for the  unknown real-world measure, 
and a nonempty set 
\begin{equation}
\label{def:C}
\begin{split}
\mathcal{C}&\subseteq \big\{X\colon \Omega \to [0,\infty]: X \mbox{ is } \cF^*\mbox{-measurable}\big\},
\end{split}
\end{equation}
where we set $\mathcal{C}(x):=x\mathcal{C}$ for every $x>0$. Then, we fix  $\emptyset\neq\fP\subseteq \fP(\Omega)$\footnote{typically, one chooses $\fP:=\fP(\Omega)$ or $\fP:=\{\Q \in \fP(\Omega)\colon \exists\, \P \in \cP \mbox{ such that } \Q \approx\P\}$; see also Section~\ref{sec:application} and Section~\ref{sec:proof-app}.} 
and define
\begin{equation}
\label{def:D}
\begin{split}
\mathcal{D}&:= \big\{\Q \in \fP\colon \E_\Q[X]\leq 1 \mbox{ for all } X \in \mathcal{C}\big\},
\end{split}
\end{equation}
where we set $\mathcal{D}(y):=y\mathcal{D}$ for every $y>0$. In other words, we interpret $\mathcal{D}$, roughly speaking, as the polar set of $\mathcal{C}$. 
%

%
%
We impose the following assumption,
which is 
standard in the utility maximisation literature in mathematical finance (cf., e.g., \cite{KramkovSchachermayer.99}).
%
%
\begin{assumption}\label{ass:no-arbitrage-style}
	For every $\P \in \cP$ there exists a $\Q \in \mathcal{D}$ such that $\Q\ll \P$.
\end{assumption}

We call a function $U:(0,\infty)\to[-\infty,\infty)$ a utility function if it is concave and non-decreasing. 
We fix a utility function $U:(0,\infty)\to[-\infty,\infty)$ and define $U(0):=\lim_{x \downarrow 0} U(x)$. 
%
%
%
%
Moreover, we consider the conjugate function 
\begin{equation}
\begin{split}
V(y)&:=\sup_{x\geq 0}  \big[U(x)-xy\big], \quad y>0,\\
V(0)&:=\lim_{y \downarrow 0} V(y),\\
V(y)&:= + \infty, \quad y<0.\\
	\end{split}
	\end{equation}
%
%
In the following two subsections we present our main duality results for the robust utility maximisation problem in the abstract setting, where we distinguish the two cases where $U(0)>-\infty$ and $U(0)=-\infty$.
\subsection{Main result for utility functions bounded from below}\label{subsec:abstract-boundedbelow}
In this subsection, we provide our main result for utility functions which satisfy that $U(0)>-\infty$. More precisely, in this subsection, we impose the following condition on the utility function.
\begin{assumption}
	\label{ass:U-1}
	The utility function $U:[0,\infty)\to \R$ is real valued.
\end{assumption}
The robust utility maximisation problem is then defined as the following maximisation problem
\begin{equation}\label{def:u}
u(x):=\sup_{g \in \mathcal{C}(x)}\inf_{\P\in \cP} \E_\P[U(g)], \quad x>0.
\end{equation}
\begin{remark}\label{rem:U-trivial}
	Note that  Assumption~\ref{ass:U-1} ensures that $U\colon[0,\infty)\to \R$ is continuous; see \cite[Theorem~10.1, p.82]{Rockafellar.97}.
	Moreover, common utility functions defined on $(0,\infty)$ like the
	power utilities $U(x):=\frac{1}{p}x^p$, $p\in (0,1)$, and the exponential utilities $U(x)=-e^{-\lambda x}$, $\lambda>0$,   satisfy Assumption~\ref{ass:U-1}.
\end{remark}
We define the  corresponding dual function
\begin{equation}
\label{def:v}
v(y):=\inf_{\Q\in \mathcal{D}(y)}\inf_{\P\in \cP} \E_\P\big[V(\tfrac{d\Q}{d\P})\big], \quad y>0,
\end{equation}
where we make the convention $\frac{d\Q}{d\P}:=-\infty$ if $\Q$ is not absolutely continuous with respect to $\P$. 
We impose the following standard condition (see, e.g., in \cite{KramkovSchachermayer.99}) on the robust utility maximisation problem.
\begin{assumption}\label{ass:u-finite}
	There exists $x_0 \in (0,\infty)$ such that $u(x_0)<\infty$.
\end{assumption}
For the second part of the  main result of this subsection, we assume that medial limits exist.

\begin{assumption}\label{ass:limmed}
	There exists a positive linear functional $\mathop{\mathrm{lim\,med}}\colon \ell^\infty\to \R$, called \textit{medial limit}, satisfying $\liminf_{n\to \infty}\leq \mathop{\mathrm{lim\,med}}_{n\to \infty}\leq \limsup_{n\to \infty}$ so that for any uniformly bounded sequence of universally measurable functions $X_n\colon \Omega \to \R$, $n \in \N$, the medial limit $X:=\mathop{\mathrm{lim\,med}}_{n\to \infty} X_n$ is universally measurable and satisfies that $\E_\P[X]=\mathop{\mathrm{lim\,med}}_{n\to \infty}\E_{\P}[X_n]$ for every $\P \in \fP(\Omega)$.
	%
	 Moreover, following \cite{bartl2019limmed}, we extend the definition of the $\mathop{\mathrm{lim\,med}}$ from $\ell^\infty$ to $[-\infty,\infty]^\N$ by setting 
	\begin{equation*}
	\mathop{\mathrm{lim\,med}}_{n \to \infty} x_n
	:=
	\sup_{k \in \N} \inf_{m \in \N}
		\mathop{\mathrm{lim\,med}}_{n \to \infty} \big[(-m)\vee (x_n\wedge k)\big].
	\end{equation*}

\end{assumption}

\begin{remark}\label{rem:limmed}
	The existence of the medial limit is guaranteed under the usual ZFC axioms together with Martin's axiom; see \cite{meyer1973limites,normann1976martin}. In the literature of robust mathematical finance, the usage of medial limits appeared first in \cite{nutz2012pathwise} to construct an (aggregated) stochastic integral simultaneously under a set of non-dominated probability measures. In \cite{Nutz.13}, medial limits were applied to construct superhedging strategies in the quasi-sure setting in discrete-time. Moreover, in \cite{bartl2019limmed},   medial limits were used in the context of robust utility maximisation on the real line in the discrete-time setting. Roughly speaking,  medial limits turn out to be particularly useful in the robust finance theory when dealing with a set of non-dominated probability measures, as then classical limit-arguments like
	the Koml\'os theorem cannot be applied; we refer to \cite{nutz2012pathwise,Nutz.13,bartl2019limmed} for more details and properties regarding medial limits.
\end{remark}

Under the condition that Assumption~\ref{ass:limmed} holds, we can then define for every $x>0$,
\begin{equation*}
\ov{\mathcal{C}(x)}:= \Big\{\mathop{\mathrm{lim\,med}} g_n:\Omega \to [0,\infty] \colon (g_n)_{n \in \N}\subseteq \mathcal{C}(x)\Big\}
\end{equation*}
and $\ov{\mathcal{C}}\equiv\ov{\mathcal{C}(1)}$.
Observe that 
$\ov{\mathcal{C}(x)}=x\ov{\mathcal{C}}$ for every $x>0$, and  we denote it by $\ov{\mathcal{C}}(x)$.
Moreover, we also consider the following robust utility maximisation problem
\begin{equation}\label{def:u-bar}
\ov u(x):=\sup_{g \in \ov{\mathcal{C}}(x)}\inf_{\P\in \cP} \E_\P[U(g)], \quad x>0,
\end{equation}
and impose the following conditions.
\begin{assumption}\label{ass:u-bar-finite}
	There exists $x_0 \in (0,\infty)$ such that $\ov u(x_0)<\infty$.
\end{assumption}
\begin{assumption}\label{ass:U-I}
	For every $x\in (0,\infty)$ and $(g_n)_{n\in \N}\subseteq {\mathcal{C}}(x)$, the sequence of random variables
	\begin{equation*}
	\max\big\{U\big(g_n+\tfrac{1}{n}\big),0\big\}, \ n \in \mathbb N, 
	\end{equation*}
	is uniformly integrable with respect to $\P$ for all $\P \in \cP$.
\end{assumption}
\begin{remark}\label{rem:ass:U-I}
	Whereas Assumptions~\ref{ass:u-finite}~\&~\ref{ass:u-bar-finite} are standard in the mathematical finance literature, Assumption~\ref{ass:U-I} is not a common one. However, note that  every utility function $U$ which is bounded from above automatically satisfies Assumptions~\ref{ass:u-finite},~\ref{ass:u-bar-finite},~\&~\ref{ass:U-I}, no matter what $\mathcal{C}$ and $\cP$ are. 
	In addition, we show in Subsection~\ref{subsec:proof-co} that in the setting of drift and volatility uncertainty (see Section~\ref{sec:application}) the Assumptions~\ref{ass:u-finite},~\ref{ass:u-bar-finite},~\&~\ref{ass:U-I} are automatically satisfied for the logarithm, power, and exponential utility functions.
\end{remark}
Now we are ready to state our main result in the abstract setting for utility functions $U$ satisfying Assumption~\ref{ass:U-1}, which can be seen as the robust version of the classical result of Kramkov~\&~Schachermayer~\cite[Theorem~3.1]{KramkovSchachermayer.99}.
\begin{theorem}
	\label{thm:main}
	Let $U$ be a utility function satisfying Assumption~\ref{ass:U-1}, let $\mathcal{C}$, $\mathcal{D}$ be defined as in \eqref{def:C} and \eqref{def:D}, and let $\cP$ be a set of probability measures such that Assumption~\ref{ass:no-arbitrage-style} and Assumption~\ref{ass:u-finite} hold.
	Moreover, assume that 
\begin{enumerate}[(1)]
	\item\label{thm:ass-1} the set of probability measures $\cP$, $\cD$ are both convex and compact,
	\item\label{thm:ass-2} we have that 
	\begin{equation*}
	\mathcal{D}= \big\{\Q \in \fP\colon \E_\Q[X]\leq 1 \mbox{ for all } X \in (\mathcal{C}\cap C_b)\big\},
	\end{equation*}
	\item\label{thm:ass-3} we have that
	\begin{equation*}
	\big\{X \in C_b^+\colon \E_\Q[X]\leq 1 \mbox{ for all } \Q \in \mathcal{D}\big\}=\mathcal{C}\cap C_b.
	\end{equation*}
\end{enumerate}
Then the following holds:
\begin{enumerate}[(i)]
	\item\label{thm:res-1} $u$ is nondecreasing, concave, 
	and $u(x) \in \R$ for all $x>0$,
	\item\label{thm:res-1b} $v$ is nonincreasing, convex, and proper, 
	
	\item\label{thm:res-2}  the functions $u$ and  $v$ defined in \eqref{def:u} and \eqref{def:v} are conjugates, i.e.\
	\begin{equation*}
	\begin{split}
	u(x)&=\inf_{y\geq 0}\big[v(y)+xy\big], \quad x>0,\\ 
	v(y)&=\sup_{x\geq 0}\big[u(x)-xy\big], \quad y>0,
	\end{split}
	\end{equation*}
\item\label{thm:res-4-neu}  for every $x>0$ we have that
		\begin{equation*}
	u(x)\equiv \sup_{g\in {\mathcal{C}}(x)}\inf_{\P\in\mathcal{P}} \E_\P[U(g)]
	= \sup_{g\in (\mathcal{C}(x)\cap C_b)}\inf_{\P\in\mathcal{P}} \E_\P[U(g)].
	\end{equation*}
\end{enumerate}
If in addition, we assume that Assumption~\ref{ass:limmed},  Assumption~\ref{ass:u-bar-finite}, and Assumption~\ref{ass:U-I} hold, then we additionally obtain that
	\begin{enumerate}[(i)]
		\addtocounter{enumi}{4}
		
	\item\label{thm:res-4}  for every $x>0$ we have that
	\begin{equation*}
	u(x)= \ov u(x),
	\end{equation*}	
\item\label{thm:res-3}  for every $x>0$ there exists $\widehat{g} \in \ov{\mathcal{C}}(x)$ such that
	\begin{equation*}
	\inf_{\P\in\mathcal{P}} \E_\P[U(\widehat{g})]
	=\sup_{g\in \ov{\mathcal{C}}(x)}\inf_{\P\in\mathcal{P}} \E_\P[U(g)]
	\equiv\ov u(x).
	\end{equation*}
\end{enumerate}
\end{theorem}	
\begin{remark}\label{rem:v-finite}
Item~\eqref{thm:res-1b}	implies that $v(y)>-\infty$ for all $y\in [0,\infty)$ and that there exists $y_0\in (0,\infty)$ such that $v(y) \in \R$ for all $y\geq y_0$.
\end{remark}

%
%
\subsection{Main result for utility functions unbounded from below}\label{subsec:abstract-unboundedbelow}
In this section, we provide our main result for utility functions which satisfies that $U(0)=-\infty$. More precisely, in this subsection, we impose the following condition on the utility function.
\begin{assumption}
	\label{ass:U-2}
	The utility function $U:[0,\infty)\to[-\infty,\infty)$ satisfies that $U(0)=-\infty$, $U|_{(0,\infty)}$   is real valued,	and for every sequence $(x_n)_{n\in \N}\subseteq (0,\infty)$ with $\lim_{n \to \infty} x_n=\infty$, we have that
	\begin{equation*}
	\lim_{n\to \infty}\tfrac{U(x_n)}{x_n}=0.
	\end{equation*}
\end{assumption}
\begin{remark}\label{rem:U-trivial-2}
	Note that  Assumption~\ref{ass:U-2} and that $U(0)\equiv\lim_{x\downarrow0} U(x)$ ensure that $U\colon[0,\infty)\to [-\infty,\infty)$ is continuous; see \cite[Theorem~10.1, p.82]{Rockafellar.97}. 
	Moreover, common utility functions defined on $(0,\infty)$ like the 
	the logarithmic utility $U(x):=\log(x)$ and 
	the
	power utilities $U(x):=\frac{1}{p}x^p$, $p\in (-\infty,0)$,  satisfy Assumptions~\ref{ass:U-2}.
\end{remark}
In this subsection, we impose that medial limits exist (see Assumption~\ref{ass:limmed}) and consider the robust utility maximisation problem $\ov u$ defined in \eqref{def:u-bar}. 
In addition, we assume the following.
\begin{assumption}\label{ass:v-finite}
	Let 
	\begin{equation*}
	\begin{split}
	V_1(y):=\sup_{x\geq 0}  \big[U_1(x)-xy\big], \quad y>0,
	\end{split}
	\end{equation*}
	where $U_1(\cdot):=U(\cdot + 1)$, $x\geq 0$. 
	Then for each $y>0$ and each $\P \in \cP$ there exists $\Q \in \mathcal{D}$ such that
	\begin{equation*}
	\E_\P\big[\!\max\!\big\{V_1(y\tfrac{d\Q}{d\P}),0\big\}\big]
	<\infty.
	\end{equation*}
\end{assumption}
\begin{remark}\label{rem:ass:U-I-2}
	Although Assumption~\ref{ass:v-finite} is a priori not standard in the literature, we observe that it is a modest assumption. Indeed,  every utility function $U$ which is bounded from above automatically satisfies  Assumption~\ref{ass:v-finite}, no matter what $\mathcal{C}$ and $\cP$ are (note that $V$ is nonincreasing with $V_1(0)=U_1(\infty)=U(\infty)$). In addition, we show in Subsection~\ref{subsec:proof-co} that in the setting of  Section~\ref{sec:application}, Assumption~\ref{ass:v-finite} is automatically satisfied for the logarithm, power, and exponential utility functions.
	Furthermore, Assumption~\ref{ass:v-finite} implies that $v(y)<\infty$ for all $y> 0$.
\end{remark}
Then we obtain  the following result.
\begin{theorem}
	\label{thm:main-2}
	Let Assumption~\ref{ass:limmed} hold.
	Let $U$ be a utility function satisfying Assumption~\ref{ass:U-2}, let $\mathcal{C}$, $\mathcal{D}$ be defined as in \eqref{def:C} and \eqref{def:D}, and let $\cP$ be a set of probability measures such that
	Assumption~\ref{ass:no-arbitrage-style},
	Assumption~\ref{ass:u-bar-finite}, Assumption~\ref{ass:U-I}, 
	and Assumption~\ref{ass:v-finite} hold.
	Moreover, assume that 
	\begin{enumerate}[(1)]
		\item\label{thm:ass-1-2} the set of probability measures $\cP$, $\cD$ are both convex and compact,
		\item\label{thm:ass-2-2} we have that 
		\begin{equation*}
		\mathcal{D}= \big\{\Q \in \fP\colon \E_\Q[X]\leq 1 \mbox{ for all } X \in (\mathcal{C}\cap C_b)\big\},
		\end{equation*}
		\item\label{thm:ass-3-2} we have that
		\begin{equation*}
		\big\{X \in C_b^+\colon \E_\Q[X]\leq 1 \mbox{ for all } \Q \in \mathcal{D}\big\}=\mathcal{C}\cap C_b.
		\end{equation*}
	\end{enumerate}
	Then the following holds:
	\begin{enumerate}[(i)]
		\item\label{thm:res-1-2} $\ov u$ is nondecreasing, concave, 
		and $\ov u(x) \in \R$ for all $x>0$,
		\item\label{thm:res-1b-2} $v$ is nonincreasing, convex, and $v(y) \in \R$ for all $y>0$, 
		
		\item\label{thm:res-2-2}  the functions $\ov u$ and  $v$ defined in \eqref{def:u-bar} and \eqref{def:v} are conjugates, i.e.\
		\begin{equation*}
		\begin{split}
		\ov u(x)&=\inf_{y\geq 0}\big[v(y)+xy\big], \quad x>0,\\ 
		v(y)&=\sup_{x\geq 0}\big[\ov u(x)-xy\big], \quad y>0,
		\end{split}
		\end{equation*}
		\item\label{thm:res-3-2}  for every $x>0$ there exists $\widehat{g} \in \ov{\mathcal{C}}(x)$ such that
		\begin{equation*}
		\inf_{\P\in\mathcal{P}} \E_\P[U(\widehat{g})]
		=\sup_{g\in \ov{\mathcal{C}}(x)}\inf_{\P\in\mathcal{P}} \E_\P[U(g)]
		\equiv\ov u(x).
		\end{equation*}
	\end{enumerate}
\end{theorem}	
%
%
\section{Main results under drift and volatility uncertainty}\label{sec:application}

The goal of this section is to show that the main assumptions imposed in Theorem~\ref{thm:main}~\&~\ref{thm:main-2}, namely the bipolar relation of $\cC$ and $\cD$ and the convex-compactness assumption on $\cP$ and $\cD$, are naturally fulfilled in the context of robust utility maximisation under simultaneous drift and volatility uncertainty.

To that end, in this section, 
let $\Omega=C([0,T],\R^d)$ and $\mathcal{F}$ be its Borel $\sigma$-field. 
We denote by $(S_t)_{0\leq t \leq T}$ the canonical process on $C([0,T],\R^d)$, i.e.\ $S_t(\omega)=\omega(t)$. Moreover, let
$\mathbb{F}:=(\mathcal{F}_t)_{0\leq t \leq T}$ be the raw filtration generated by the canonical process $S$, i.e.\ $\mathcal{F}_t=\sigma(S_s, s\leq t)$, and denote by $\F^*=(\mathcal{F}_t^*)_{0\leq t \leq T}$ the corresponding universal $\sigma$-field. 
%

Now consider the following sets of  Borel probability measures on $\Omega$ which were introduced in \cite{neufeld2014measurability}.
\begin{equation*}
\begin{split}
&\cP_{sem}:=\big\{ \P \in \fP(\Omega)\colon S \mbox{ is a semimartingale on } (\Omega,\cF,\F,\P)\big\},\\
&\cP^{ac}_{sem}:=\big\{ \P \in \cP_{sem}\colon B^\P\ll dt, C^\P\ll dt  \ \ \P\mbox{-a.s.}\big\},
\end{split}
\end{equation*}
where $B^\P$ and $C^\P$ denotes the first and second characteristic of the continuous semimartingale $S$ under $\P$; we refer to \cite{JacodShiryaev.03} for further discussions regarding semimartingale theory. Given any Borel set $\Theta\subseteq \R^d\times\S^d_+$ we then define the set $\cP$ 
by
\begin{equation}\label{P-theta}
\cP\equiv\cP^{ac}_{sem}(\Theta):=\big\{ \P \in \cP^{ac}_{sem} \colon (b^\P, c^\P) \in \Theta \ \ \P\otimes dt\mbox{-a.e.}\big\},
\end{equation}
where $(b^\P,c^\P)$ denotes the differential characteristics of $S$ under $\P$; see also \cite{JacodShiryaev.03,neufeld2014measurability}. We use the standard notion to say  that a property holds $\cP$-q.s.\ if it holds true $\P$-a.s.\ for all $\P\in \cP$.
%
%
%
%
During this section, we fix a set $\Theta \subseteq \R^d \times \S^d_+$ and impose the following conditions.
\begin{assumption}\label{ass:invertible}
	The set $\Theta \subseteq \R^d \times \S^d_+$ satisfies the following:
	
	\vspace{0.2cm}
	\noindent
	$\bullet$ $\Theta \subseteq \R^d \times \S^d_+$ is convex and compact,\\
	$\bullet$ there exists $\underline{c} \in \S^d_{++}$ such that $\underline{c}\leq c$ for all $c \in \proj_c(\Theta)
	=:\Theta_c$, where
	\begin{equation*}
	\proj_c(\Theta):=\big\{c\in \S^d_+\colon \exists\, b \in \R^d \mbox{ such that } (b,c)\in \Theta \big\}.
	\end{equation*}	
\end{assumption}
\begin{remark} \label{rem:Theta}
	Assumption~\ref{ass:invertible} guarantees that each $c \in \Theta_c$ is in $\S^d_{++}$, in particular is invertible. Moreover, we have for each $c \in \Theta_c$ that both $c$ and $c^{-1}$ are bounded. The uniform ellipticity condition in Assumption~\ref{ass:invertible} however imposes that each $\P \in \cP$ corresponds to a complete financial market. From a technical point of view, this condition allows for each $\P \in \cP$ to guarantee the existence of an equivalent martingale measure $\Q \in  \cM := \cP^{ac}_{sem}(\widetilde{\Theta})$, where $\widetilde\Theta:=\{0,\dots,0\}\times \Theta_c\subseteq \R^d \times \S^d_+$, and vice versa for each $\Q \in \cM$ that there exists $\P \in \cP$ such that $\P \approx \Q$; we refer to Proposition~\ref{le:Q-P}. We point out that a similar condition has also been imposed in Denis~\&~Kervarec \cite{DenisKervarec.13}; see Hypothesis (H) in their paper. This in turn allows us to identify $\cD=\cM$, which together with Proposition~\ref{prop:separation-app} is the key property enabling us to show that the bipolar relation on the subset $C_b$ introduced in \eqref{intro:bipolar-our-1} and \eqref{intro:bipolar-our-2} naturally holds in the context of drift and volatility uncertainty; we refer to Proposition~\ref{prop:D-C-Cb-polar} and Proposition~\ref{prop:C-C-b-char}.
\end{remark}
Next, let us introduce a particular filtration $\mathbb G:=(\mathcal{G}_t)_{0\leq t\leq T}$ defined by 
\begin{equation}\label{def:filtration-G}
\begin{split}
\mathcal{G}_t&:= \bigcap_{s>t} \big(\cF^*_s \vee \mathcal{N}^\cP\big), 
\qquad 0 \leq t \leq T,
\end{split}
\end{equation} 
where $\mathcal{N}^{\cP}$ is the collection of all sets which are $\cF_T$-$\P$-null for all $\P \in \cP$.  A priori, the filtration $\mathbb G$ looks non-natural. However,  it will be helpful in the sequel to apply results in \cite{NeufeldNutz.12,nutz2015robust} where this filtration has been used; see also Remark~\ref{rem:Nullsets}.
 In addition, note that for every $\P \in \cP$, the filtration $\mathbb G$ satisfies that $\F\subseteq \mathbb{G} \subseteq \F^\P_+$, 
where $\F_+$ denotes the right continuous version of $\F$ and  $\F_+^\P$ denotes the usual $\P$-augmentation of $\F$; see also the following Remark~\ref{rem:filtration}.
\begin{remark}\label{rem:filtration}
	By \cite[Proposition~2.2]{neufeld2014measurability} we know that $(S_t)_{0\leq t \leq T}$ is a $\P$-$\F$-semimartingale if and only if it is a $\P$-$\F_+$-semimartingale, as well as if and only if it is a $\P$-$\F_+^\P$-semimartingale.
	 Moreover, the associated semimartingale characteristics with respect to these filtrations are the same. In particular, we see that \eqref{P-theta} does not depend on the 
	choice of the filtration $\mathbb G$, as long as $\F \subseteq \mathbb G \subseteq \F^\P_+$.
\end{remark}

Furthermore, for any fixed  $\P\in \fP(\Omega)$ such that $S$ is a $\P$-semimartingale and any (predictable) process $H$ which is $\P$-$S$-integrable in the semimartingale sense (see, e.g., \cite[Definition~III.6.17, p.207]{JacodShiryaev.03}), we denote by $\int H\,dS\equiv(H\cdot S)\equiv ^{(\P)}\!\!(H\cdot S)$ the usual stochastic integral under $\P$. 
Then, we define $\mathcal{H}$ to be the set of all $\mathbb G$-predictable processes $H$ which are $\P$-$S$-integrable in the semimartingale sense 
for all $\P \in \cP$  such that 
$(H\cdot S)\geq -c$ \, $\P$-a.s.\ $\forall\, \P\in \cP$ 
for some constant $c>0$,
where $c$ may depend on $H$ and $\mathbb{P}$.
%
%
%
%
Finally we specify the sets $\cC$, $\cD$ appearing in Theorem~\ref{thm:main} of the previous section. We define\footnote{We choose $\fP:=\fP_e(\cP)$ in the definition of $\cD$.}
\begin{equation}\label{C-D-applic}
\begin{split}
\cC&\!:=\!\Big\{ X: \Omega \to [0,\infty] \mbox{ }\cF^*_T\mbox{-measurable}\colon \exists\, H \in \cH \mbox{ so that } 1+ (H \cdot S)_T \geq  X \  \,
\cP\mbox{-q.s.}
\Big\},\\
\fP_e(\cP)&:=\Big\{\Q \in \fP(\Omega)\colon  \exists\, \P \in \cP \mbox{ such that } \Q \approx \P
\Big\},\\
\cD&:=
\Big\{\Q \in \fP_e(\cP)\colon \E_\Q[X]\leq 1 \ \mbox{ for all } X \in \cC \Big\}.
\end{split}
\end{equation}
%
%
Moreover, for every $x,y>0$ we define the sets $\cC(x)$, $\ov \cC(x)$, and $\cD(y)$, as well as the functions $u(x)$, $\ov u(x)$, and $v(y)$ analog to Section~\ref{sec:abstract}.
Now we are able to state the main results of this section.
We distinguish the two cases where $U(0)>-\infty$ and $U(0)=-\infty$.
%
%
%
%

Now, we provide our main result under the setting of Section~\ref{sec:application} for utility functions which satisfy that $U(0)>-\infty$. 
\begin{theorem}
	\label{thm:main-app}
	Let $U$ be a utility function satisfying Assumption~\ref{ass:U-1}, let $\cP$, $\mathcal{C}$, and $\mathcal{D}$ be defined as in  \eqref{P-theta} and \eqref{C-D-applic} such that  Assumption~\ref{ass:u-finite} and Assumption~\ref{ass:invertible} hold.
	Then 
	\begin{enumerate}[(I)]
		\item 
		Item~(\ref{thm:res-1}), Item~(\ref{thm:res-1b}), Item~(\ref{thm:res-2}), and Item~(\ref{thm:res-4-neu}) of Theorem~\ref{thm:main} hold.
	\end{enumerate}
If in addition, we assume that Assumption~\ref{ass:limmed}, Assumption~\ref{ass:u-bar-finite}, and Assumption~\ref{ass:U-I} hold, then we additionally obtain that
\begin{enumerate}[(I)]
	\addtocounter{enumi}{1}
	
		\item Item~(\ref{thm:res-4}) and Item~(\ref{thm:res-3}) of Theorem~\ref{thm:main} hold.
	\end{enumerate}
\end{theorem}
%
%

%
%

Next, we provide our main result under the setting of Section~\ref{sec:application} for utility functions which satisfy that $U(0)=-\infty$. 
\begin{theorem}
	\label{thm:main-app-2}
	Let Assumption~\ref{ass:limmed} hold.
	Let $U$ be a utility function satisfying Assumption~\ref{ass:U-2}, let $\cP$, $\mathcal{C}$, and $\mathcal{D}$ be defined as in  \eqref{P-theta} and \eqref{C-D-applic} such that
	Assumptions~\ref{ass:u-bar-finite},~\ref{ass:U-I},~\&~\ref{ass:v-finite},
	 and Assumption~\ref{ass:invertible} hold.
	Then 
	\begin{enumerate}[(I)]
		\item  Item~(\ref{thm:res-1-2}), Item~(\ref{thm:res-1b-2}), Item~(\ref{thm:res-2-2}), and Item~(\ref{thm:res-3-2}) of Theorem~\ref{thm:main-2} hold.
	\end{enumerate}
\end{theorem}
The  idea of the proof of Theorem~\ref{thm:main-app} and Theorem~\ref{thm:main-app-2} is to verify 
the bipolar relation of $\cC$ and $\cD$ and the convex-compactness assumption on $\cP$ and $\cD$ to be able to apply Theorem~\ref{thm:main} and Theorem~\ref{thm:main-2}, respectively.
We refer to Subsection~\ref{subsec:proof-app} for their proofs. 
%
%

Finally, we would like to emphasise that Assumptions~\ref{ass:u-finite}~\&~\ref{ass:u-bar-finite} and Assumptions~\ref{ass:U-I}~\&~\ref{ass:v-finite}  are naturally satisfied in the setting of Section~\ref{sec:application}  by showing that they automatically hold true in the case where   $U(x)=\log(x)$, $U(x)=\tfrac{x^p}{p}$, $p \in (-\infty,0)\cup(0,1)$, and $U(x)=-e^{-\lambda x}$, $\lambda>0$ (see also Remark~\ref{rem:ass:U-I} and Remark~\ref{rem:ass:U-I-2}). As in the previous results 
we distinguish the two cases where $U(0)>-\infty$ and $U(0)=-\infty$.  
\begin{corollary}
	\label{co:main-app}
	Let $U$ be either  a power utility $U(x)=\tfrac{x^p}{p}$ for some $p\in (0,1)$, or an exponential utility function $U(x)=-e^{-\lambda x}$ for some $\lambda>0$. Moreover, let $\cP$, $\mathcal{C}$, and $\mathcal{D}$ be defined as in  \eqref{P-theta} and \eqref{C-D-applic} such that  Assumption~\ref{ass:invertible} holds.
	Then 
	\begin{enumerate}[(I)]
		\item Item~(\ref{thm:res-1}), Item~(\ref{thm:res-1b}), Item~(\ref{thm:res-2}), and Item~(\ref{thm:res-4-neu}) of Theorem~\ref{thm:main} hold, and $v(y) \in \R$ for all $y>0$.
	\end{enumerate}
If in addition, we assume that Assumption~\ref{ass:limmed} holds, then we additionally obtain that
\begin{enumerate}[(I)]
	\addtocounter{enumi}{1}
		\item Item~(\ref{thm:res-4}) and Item~(\ref{thm:res-3}) of Theorem~\ref{thm:main} hold.
\end{enumerate}
\end{corollary}
The following corollary corresponds to the case where $U(0)=-\infty$.
\begin{corollary}
	\label{co:main-app-2}
	Let Assumption~\ref{ass:limmed} hold.
	Let $U$ be either the log utility function $U(x)=\log(x)$, or a power utility $U(x)=\tfrac{x^p}{p}$ for some $p\in (-\infty,0)$. Moreover, let $\cP$, $\mathcal{C}$, and $\mathcal{D}$ be defined as in  \eqref{P-theta} and \eqref{C-D-applic} such that  Assumption~\ref{ass:invertible} holds.
	Then 
	\begin{enumerate}[(I)]
		\item Item~(\ref{thm:res-1-2}), Item~(\ref{thm:res-1b-2}), Item~(\ref{thm:res-2-2}), and Item~(\ref{thm:res-3-2}) of Theorem~\ref{thm:main-2} hold.
	\end{enumerate}
\end{corollary}
The proof of Corollary~\ref{co:main-app} and Corollary~\ref{co:main-app-2} are provided in Subsection~\ref{subsec:proof-co}.
%
\section{Proof of Theorem~\ref{thm:main}  and Theorem~\ref{thm:main-2}}\label{sec:proof}
We first start with two well-known results on the extension of utility functions defined on $(0,\infty)$, which we provide for the sake of completeness.
\begin{lemma}\label{le:U-extension}
Let $U\colon(0,\infty)\to \R$ be nondecreasing and concave. 
Let $V\colon(0,\infty) \to (-\infty,\infty]$ be defined by
\begin{equation}\label{eq:le:V}
\begin{split}
V(y)&:=\sup_{x\geq 0}\big[U(x)-xy\big], \quad y>0.\\
\end{split}
\end{equation}
Moreover, define  $\widetilde{U}\colon\R \to [-\infty,\infty)$ and $\widetilde{V}\colon\R \to (-\infty,\infty]$ by
\begin{equation}\label{eq:U-V-extension}
\widetilde{U}(x):=\begin{cases}
U(x) 	& 	x> 0,\\
\lim_{x\downarrow 0} U(x) & x=0,\\
-\infty & x<0,
\end{cases}
\qquad \quad \mbox{ and } \qquad \quad
\widetilde{V}(y):=\begin{cases}
V(y) 	& 	y> 0,\\
\lim_{y\downarrow 0} V(y) & y=0,\\
\infty & y<0.
\end{cases}
\end{equation}
Furthermore, define the function $\varphi\colon\R \to (-\infty,\infty]$ by $\varphi(x)=-\widetilde{U}(-x)$, $x\in \R$.
Then 
\begin{enumerate}[(i)]
	\item\label{U-exten-1} $\widetilde{U}$ is nondecreasing, concave, proper, upper-semicontinuous;
	
	\item\label{U-exten-2}  $\widetilde{V}$ is nonincreasing, convex, proper, lower-semicontinuous;
	
	\item\label{U-exten-3}  $\widetilde{V}$ is the convex conjugate of $\varphi$;
	\item \label{U-exten-4} We have for every $x>0$ that
	\begin{equation*}
	U(x)=\inf_{y\geq 0}\big[V(y)+xy\big].
	\end{equation*}
\end{enumerate}
\end{lemma}
\begin{proof}
Note that Item~\eqref{U-exten-1} and that $\widetilde{V}$ is nonincreasing follows directly from the definitions and the assumptions imposed on these function together with \cite[Theorem~10.1, p.82]{Rockafellar.97}. As a consequence, $\varphi$ is convex, proper, and lower-semicontinuous. Hence  the biconjugate theorem (see \cite[Theorem~12.2, p.104]{Rockafellar.97} and \cite[p.52]{Rockafellar.97}) ensures that the conjugate $\varphi^*$ of $\varphi$ is convex, proper, lower-semicontinuous and that $\varphi^{**}=\varphi$. Therefore, to prove Item~\eqref{U-exten-2} and Item~\eqref{U-exten-3}, it remains to show that $\varphi^* =\widetilde{V}$.

To that end, note that
 \eqref{eq:U-V-extension} implies 
for every $y \in \R $ that
\begin{equation*}
\begin{split}
\varphi^*(y)
=\sup_{x \in \R} \big[xy- (-\widetilde{U}(-x))\big]
&=
\sup_{x \in \R} \big[-xy+\widetilde{U}(x)\big]=
\sup_{x \geq 0} \big[-xy+\widetilde{U}(x)\big].
\end{split}
\end{equation*}
As a consequence, we see that $\R \ni y \mapsto\varphi^*(y)$ is nonincreasing, that for any $y<0$,
\begin{equation}\label{v-tilde-infty}
\varphi^*(y) =
\sup_{x \geq 0} \big[-xy+\widetilde{U}(x)\big] 
=\sup_{x \geq 0} \big[x|y|+\widetilde{U}(x)\big]
=+ \infty,
\end{equation}
and due to \eqref{eq:le:V} that for any $y> 0$,
\begin{equation}\label{v-tilde-v}
\varphi^*(y) =
\sup_{x \geq 0} \big[-xy+\widetilde{U}(x)\big] 
=
\sup_{x \geq 0} \big[-xy+U(x)\big]
=V(y).
\end{equation}
Moreover, observe that  $\varphi^*$ being nonincreasing implies that $\varphi^*(0)\geq \limsup_{y \downarrow 0} \varphi^*(y)$, whereas the lower-semicontinuity implies that 
$\varphi^*(0)\leq \liminf_{y\downarrow 0} \varphi^*(0)$.
Therefore, we obtain by \eqref{v-tilde-v} that
\begin{equation}\label{v-tilde-v-0}
\widetilde V(0)=\lim_{y\downarrow 0} V(y) = \lim_{y \downarrow 0} \varphi^*(y) = \varphi^*(0).
\end{equation}
This shows that $\varphi^*=\widetilde{V}$.

Finally, to see that Item~\eqref{U-exten-4} holds,
note that
the biconjugate theorem (see \cite[Theorem~12.2, p.104]{Rockafellar.97}) 
and Item~\eqref{U-exten-3}
imply
that 
\begin{equation*}
\sup_{y \in \R}\big[xy-\widetilde V(y)]
=\varphi^{**}(x)=\varphi(x)=-\widetilde{U}(-x), \quad x \in \R.
\end{equation*}
Therefore, we deduce  from \eqref{eq:U-V-extension} 
that for all $x>0$,
\begin{equation*}
\begin{split}
U(x)
=\widetilde{U}(x)
&=  -\sup_{y \in \R}\big[-xy-\widetilde V(y)]
= \inf_{y \in \R}\big[xy+\widetilde V(y)]
= \inf_{y \geq 0}\big[xy+\widetilde V(y)]
= \inf_{y \geq 0}\big[xy+ V(y)].
\end{split}
\end{equation*}
\end{proof}
%
\begin{lemma}\label{le:V-extension}
	Let $V\colon(0,\infty)\to \R$ be nonincreasing and convex. 
	Let $U\colon(0,\infty) \to [-\infty,\infty)$ be defined by
	\begin{equation}\label{eq:le:U}
	\begin{split}
	U(x)&:=\inf_{y\geq 0}\big[V(y)+xy\big], \quad x>0.
	\end{split}
	\end{equation}
	Moreover, define  $\widetilde{U}\colon\R \to [-\infty,\infty)$ and $\widetilde{V}\colon\R \to (-\infty,\infty]$ by
	\begin{equation*}
	\widetilde{U}(x):=\begin{cases}
	U(x) 	& 	x> 0,\\
	\lim_{x\downarrow 0} U(x) & x=0,\\
	-\infty & x<0,
	\end{cases}
	\qquad \quad \mbox{ and } \qquad \quad
	\widetilde{V}(y):=\begin{cases}
	V(y) 	& 	y> 0,\\
	\lim_{y\downarrow 0} V(y) & y=0,\\
	\infty & y<0.
	\end{cases}
	\end{equation*}
	Furthermore, define the function $\varphi\colon\R \to (-\infty,\infty]$ by $\varphi(x)=-\widetilde{U}(-x)$, $x\in \R$.
	Then 
	\begin{enumerate}[(i)]

		\item\label{V-exten-2}  $\widetilde{U}$ is nondecreasing, concave, proper, upper-semicontinuous;
		
			\item\label{V-exten-1} $\widetilde{V}$ is nonincreasing, convex, proper, lower-semicontinuous;
		
		\item\label{V-exten-3}  $\widetilde{V}$ is the convex conjugate of $\varphi$;
		\item \label{V-exten-4} We have for every $y>0$ that
		\begin{equation*}
		V(y)=\sup_{x\geq 0}\big[U(x)- xy\big].
		\end{equation*}
	\end{enumerate}
\end{lemma}
\begin{proof}
	First, observe that Item~\eqref{V-exten-1} and that $\widetilde{U}$ is nondecreasing follows from their definitions and
	\cite[Theorem~10.1, p.82]{Rockafellar.97}. 
	 Moreover, since for any $y\geq0$ the function $(0,\infty)\ni x \mapsto V(y)+xy$ is continuous and affine,
	 we get by \eqref{eq:le:U} that $\widetilde{U}$ is concave and upper-semicontinuous. As a consequence, we see that $\varphi$ is a convex  lower-semicontinuous function.
	Moreover, note that \eqref{eq:le:U} and the definitions of $\widetilde{U}$, $\widetilde{V}$  imply for any $x \in \R$  that
	\begin{equation*}
 \inf_{y\in \R}\big[\widetilde{V}(y)+xy\big]=\inf_{y\geq 0}\big[\widetilde{V}(y)+xy\big]= 
 \widetilde{U}(x).
	\end{equation*}
	Therefore, we get that
	\begin{equation*}
	\begin{split}
	\varphi(x)
	=
	-\widetilde{U}(-x)
	=
	 -\inf_{y\in \R}\big[\widetilde{V}(y)-xy\big]
	 =
	 \sup_{y\in \R}\big[-\widetilde{V}(y)+xy\big].
	\end{split}
	\end{equation*}
	Hence, we conclude that $\varphi(x)$ is the convex conjugate of $\widetilde{V}$. 
	In particular, as $\widetilde{V}$ is proper,  we get from \cite[Theorem~12.2, p.104]{Rockafellar.97} that  $\varphi(x)$ and hence also $\widetilde{U}$ is proper.
Moreover, by the biconjugate theorem (see \cite[Theorem~12.2, p.104]{Rockafellar.97}), we have that
	$\widetilde V=\widetilde V^{**}=\varphi^*$.
	 Thus we see that indeed, Items~\eqref{V-exten-2}--\eqref{V-exten-3} hold.
	
	Finally, %
	Items~\eqref{V-exten-2}--\eqref{V-exten-3} and 
	the biconjugate theorem (see \cite[Theorem~12.2, p.104]{Rockafellar.97}) imply that  for all $y>0$,
	\begin{equation*}
	\begin{split}
	V(y)
	&=
	\widetilde{V}(y)
	= 
	\widetilde{V}^{**}(y)
	=\sup_{x \in \R} \big[xy-\widetilde{V}^*(x)\big]
	=\sup_{x \in \R} \big[xy-\varphi^{**}(x)\big]\\
	&=\sup_{x \in \R} \big[xy-\varphi(x)\big]
	=\sup_{x \in \R} \big[xy+\widetilde{U}(-x)\big]
	=\sup_{x \in \R} \big[-xy+\widetilde{U}(x)\big]
	=\sup_{x\geq 0} \big[-xy+\widetilde{U}(x)\big]\\
	&=\sup_{x\geq 0} \big[-xy+U(x)\big].
	\end{split}
	\end{equation*}
\end{proof}
We also consider the following robust maximisation problem, which will be useful in the sequel:
\begin{equation}
\label{eq:u-small-v}
\begin{split}
u_c(x)&:=\sup_{g\in (\mathcal{C}(x)\cap{C_b})} \inf_{\P\in\mathcal{\cP}} \E_\P[U(g)], \quad x>0.
\end{split}
\end{equation}
\begin{lemma}\label{le:u-prop}
Suppose   that $\big\{X \in C_b^+\colon \E_\Q[X]\leq 1 \mbox{ for all } \Q \in \mathcal{D}\big\}=\mathcal{C}\cap C_b$ and that Assumptions~\ref{ass:no-arbitrage-style}~\&~\ref{ass:u-finite} hold.
Then the functions $(0,\infty) \ni x \mapsto u(x)$ and $(0,\infty) \ni x \mapsto u_c(x)$ are finite valued, nondecreasing, and concave. In particular, when setting $u(0):=\lim_{x\downarrow 0} u(x)$ and $u_c(0):=\lim_{x\downarrow 0} u_c(x)$, then both  $[0,\infty) \ni x \mapsto u(x)$ and $[0,\infty) \ni x \mapsto u_c(x)$ are continuous.

Moreover, if in addition Assumptions~\ref{ass:limmed}~\&~\ref{ass:u-bar-finite} hold, then the same holds true also for the function $(0,\infty) \ni x \mapsto \ov u(x)$.
\end{lemma}
\begin{proof}
	First, note that 
	the assumptions 
	ensure that the constant function 1 is in $\mathcal{C}$. This implies for every $x>0$ that
	\begin{equation}\label{eq:u-U-minusInf}
	u(x)
	=\sup_{g\in {\mathcal{C}}(x)}\inf_{\P \in \cP} \E_\P\big[U(g)\big]
	\geq \inf_{\P \in \cP} \E_\P\big[U(x)\big] =U(x)>-\infty.
	\end{equation}
Since $U$ is  concave and nondecreasing, and since ${\mathcal{C}}(x)=x{\mathcal{C}}$ for all $x>0$, it immediately follows that $u$ is concave and nondecreasing, too. 
%
%
%
Furthermore, $U$ being nondecreasing, concave and Assumption~\ref{ass:u-finite} ensure that $u(x)<\infty$ for every $x>0$. 
%
Together with \eqref{eq:u-U-minusInf} we  see that $u(x)\in \R$ for all $x>0$.
Finally, the continuity of $u$ now follows from \cite[Theorem~10.1, p.82]{Rockafellar.97}.

 Next, since $1 \in (\mathcal{C}\cap C_b) \subseteq \mathcal{C}$, the same arguments guarantee that the results also hold for $u_c$.

For the second part, note that Assumption~\ref{ass:limmed} ensures that medial limits exist and hence $\ov{u}$ is well-defined. Using  Assumption~\ref{ass:u-bar-finite}, the result for $\ov u$ now follows by the  same arguments.
\end{proof}
%
%
%
From now on, we define
\begin{equation}
\begin{split}
u(0)&:=\lim_{x\downarrow 0} u(x) \in [-\infty,\infty),\\
u_c(0)&:=\lim _{x\downarrow 0} u_c(x) \in [-\infty,\infty),\\
\ov u(0)&:=\lim_{x\downarrow 0} \ov u(x) \in [-\infty,\infty),
\end{split}
\end{equation}
which is well-defined by Lemma~\ref{le:u-prop}.
We start with the proof of our main results in Theorem~\ref{thm:main} in the easier setting that $U$ satisfies Assumption~\ref{ass:U-1}.
As we will see later, this will help us to  prove the corresponding results of Theorem~\ref{thm:main-2} in the case where $U$ satisfies Assumption~\ref{ass:U-2}.
%
%
%
%
\begin{proof}[Proof of Theorem~\ref{thm:main}]
We start to prove the first part of Theorem~\ref{thm:main} (which does not involve $\ov{u}$).
To that end, note first that Item~\eqref{thm:res-1} has been proved in Lemma~\ref{le:u-prop}.
As a next step, we prove Item~\eqref{thm:res-2} and Item~\eqref{thm:res-4-neu}.
%
Note that the definition of  $V$ ensures  for any $y>0$, $x> 0$, $g \in {\mathcal{C}}(x)$, $\P \in \cP$, $\Q \in \mathcal{D}$ with $\Q\ll \P$ that
\begin{equation}
\label{eq:weak-duality-U-V}
\begin{split}
\E_\P\big[V(y\tfrac{d\Q}{d\P})\big]
\geq 
\E_\P\big[U(g)-gy\tfrac{d\Q}{d\P}\big]
=\E_\P[U(g)]-y\E_{\Q}[g\big]
\geq \E_\P[U(g)]-xy.
\end{split}
\end{equation}
This assures for every $x,y> 0$ that
\begin{equation}\label{eq:weak-duality-U-V-2}
\sup_{g\in (\mathcal{C}(x)\cap C_b)}\inf_{\P\in\mathcal{P}} \E_\P[U(g)] -xy
\leq  \sup_{g\in {\mathcal{C}}(x)}\inf_{\P\in\mathcal{P}} \E_\P[U(g)]-xy
\leq v(y),
\end{equation}
which in turn implies that 
\begin{equation}
\label{eq:weak-duality-u-v-0}
\sup_{x> 0}\big[u_c(x)-xy\big]\leq \sup_{x> 0}\big[u(x)-xy\big]\leq v(y), \quad y>0.
\end{equation}
Moreover, \eqref{eq:weak-duality-u-v-0} implies for every $y>0$ 
\begin{equation}
\label{u-0-weak-duality}
\begin{split}
u(0)&=\lim _{x\downarrow 0} \big[u(x)-xy\big] \leq v(y),
\end{split}
\end{equation}
and hence we obtain the weak duality  
\begin{equation}
\label{eq:weak-duality-u-v}
\sup_{x\geq 0}\big[u_c(x)-xy\big]\leq\sup_{x\geq 0}\big[u(x)-xy\big]\leq v(y), \quad y>0.
\end{equation}
%
%
To see the opposite inequalities, note that by the bipolar representation in Items~\eqref{thm:ass-2}~\&~\eqref{thm:ass-3} 
it holds for every $x>0$, $y>0$, and $g \in C_b^+$ that $g \in \cC(x)\cap C_b$ if and only if $\sup_{\Q\in \cD(y)}\E_{\Q}[g]\leq xy$, and hence
we obtain 
for every $y>0$ that 
\begin{equation}
\label{eq:thm-bounded-1}
\begin{split}
\sup_{x> 0} \big[u_c(x)-xy\big]
&=
\sup_{x> 0}\sup_{g\in(\mathcal{C}(x)\cap C_b)}\inf_{\P\in\mathcal{P}} \big( \E_\P[U(g)]-xy\big) \\
&=
\sup_{g\in C_b^+} \inf_{\P\in\mathcal{P}}\inf_{\Q\in \mathcal{D}(y)} \big( \E_\P[U(g)]-\E_\Q[g]\big).
\end{split}
\end{equation} 
%
%
Now, for every $g\in C_b^+$, the mapping
\begin{equation}\label{eq:minimax-1}
\mathcal{D} \times \cP\ni(\Q,\P)\mapsto \E_\P[U(g)]-\E_\Q[g] 
\end{equation}
is convex and, since $U(g)$ is bounded from below, also lower semicontinuous.
Moreover, for every fixed $(\Q,\P)\in\mathcal{D} \times \cP$, 
the mapping 
\begin{equation*}
C_b^+\ni g\mapsto \E_\P[U(g)]-\E_\Q[g]
\end{equation*}
is concave. This, \eqref{eq:minimax-1}, and the assumption that both $\cD$ and $\cP$ are compact ensure that we can apply Sion's minimax theorem \cite[Theorem~4.2']{Sion.58} which establishes 
for every $y>0$ 
that
\begin{equation}\label{eq:minimax}
\sup_{g\in C_b^+}\inf_{\P\in\mathcal{P},\,\Q\in \mathcal{D}(y)} \big( \E_\P[U(g)]-\E_\Q[g]\big)
=\inf_{\P\in\mathcal{P},\,\Q\in \mathcal{D}(y)}\sup_{g\in C_b^+} \big( \E_\P[U(g)]-\E_\Q[g]\big).
\end{equation}
Moreover, one can check that for any fixed $\P\in \cP$ and $\Q\in\mathcal{D}(y)$  such that $\Q\ll\P$ that 
\begin{equation*}
\sup_{g\in C_b^+} \big( \E_\P[U(g)]-\E_\Q[g]\big)
= \E_\P\Big[\sup_{x> 0}\big( U(x)-\tfrac{d\Q}{d\P}x\big)\Big]
=\E_\P\big[ V\big(\tfrac{d\Q}{d\P}\big) \big].
\end{equation*}
This, \eqref{eq:thm-bounded-1}, and \eqref{eq:minimax} demonstrate that for every $y> 0$,
\begin{equation}
\label{eq:le:duality-v-c}
\sup_{x> 0} \big[u_c(x)-xy\big]=v(y).
\end{equation}
Therefore, the weak duality \eqref{eq:weak-duality-u-v} 
and the fact that $u\geq u_c$ 
imply that 
\begin{equation}
\label{eq:le:duality-v}
\sup_{x\geq 0} \big[u_c(x)-xy\big]=\sup_{x\geq 0} \big[u(x)-xy\big]=v(y), \quad y>0.
\end{equation}
Moreover, note that  \eqref{eq:le:duality-v} together with Lemma~\ref{le:u-prop} and Lemma~\ref{le:U-extension} show that
%
\begin{equation*}
u_c(x)=\inf_{y\geq 0}\big[v(y)+xy\big]=u(x), \quad x>0,
\end{equation*}
which together with \eqref{eq:le:duality-v} indeed proves that Item~\eqref{thm:res-2} and Item~\eqref{thm:res-4-neu} hold.
%
Furthermore, note that Item~\eqref{thm:res-1} and Item~\eqref{thm:res-2} together with Lemma~\ref{le:U-extension} imply that $v$ is nonincreasing, convex, and proper, 
which proves Item~\eqref{thm:res-1b}. This finishes the first part of the proof.

To prove the second part of Theorem~\ref{thm:main} (which involves $\ov u$),
note that by definition of $g \in \ov{\mathcal{C}}(x)$, there exists a sequence $(g_n)_{n \in \N} \subseteq \mathcal{C}(x)$ such that $g= \mathop{\mathrm{lim\,med}}_{n \to \infty}  g_n$. Therefore, the definition of $V$ and Fatou's lemma for the medial limit (see \cite[Lemma~3.8(v)]{bartl2019limmed}) imply that
for any $y>0$, $x> 0$, $g \in \ov{\mathcal{C}}(x)$, $\P \in \cP$, $\Q \in \mathcal{D}$ with $\Q\ll \P$,
\begin{equation}
\label{eq:weak-duality-U-V-1}
\begin{split}
\E_\P\big[V(y\tfrac{d\Q}{d\P})\big]
& \geq  \E_\P\big[U(g)\big]-\E_\P[gy\tfrac{d\Q}{d\P}\big]
\\
&=\E_\P\big[U(g)\big] -y\E_\Q\big[\mathop{\mathrm{lim\,med}}_{n \to \infty} g_n\big]\\
&\geq  \E_\P\big[U(g)\big] -y\mathop{\mathrm{lim\,med}}_{n \to \infty} \E_\Q\big[ g_n\big]\\
& \geq  \E_\P\big[U(g)\big] -xy.
\end{split}
\end{equation}
This, the fact $u(x)\leq \ov u(x)$ as $\cC(x) \subseteq \ov{\cC}(x)$ for every $x>0$, \eqref{eq:le:duality-v}, and \eqref{u-0-weak-duality} (with $u$ replaced by $\ov u$) show that
\begin{equation*}
v(y)=\sup_{x\geq 0}\big[u(x)-xy\big]\leq\sup_{x\geq 0}\big[\ov u(x)-xy\big]\leq v(y), \quad y>0,
\end{equation*}
which implies that 
\begin{equation*}
\sup_{x\geq 0} \big[u(x)-xy\big]=\sup_{x\geq 0} \big[\ov u(x)-xy\big]=v(y), \quad y>0.
\end{equation*}
Combining this with Lemma~\ref{le:u-prop} and Lemma~\ref{le:U-extension} shows that
%
\begin{equation*}
\ov u(x)=\inf_{y\geq 0}\big[v(y)+xy\big]=u(x), \quad x>0,
\end{equation*}
which proves Item~\eqref{thm:res-4}.

Finally, to see that Item~\eqref{thm:res-3} holds, 
we know from Item~\eqref{thm:res-4} that $\ov u=u$,  hence for each $n \in \N$ there exists an element $g_n\in {\mathcal{C}}(x)$ such that 
\begin{equation}\label{eq:pf_1-bounded}
\ov u(x)\leq \inf_{\P\in\mathcal{P}} \E_\P[U(g_n)]+\tfrac{1}{n}.
\end{equation}
Define 
\begin{equation*}
\widehat{g}:=\mathop{\mathrm{lim\,med}}_{n \to \infty} g_n\in \ov{\mathcal{C}}(x).
\end{equation*}
Since $U$ is concave, we obtain by Jensen's inequality for medial limits (see \cite[Lemma~3.8(iii)]{bartl2019limmed}) that
\begin{equation*}
U(\widehat{g})
=
U\big(\mathop{\mathrm{lim\,med}}_{n\to \infty}g_n\big)
\geq 
\mathop{\mathrm{lim\,med}}_{n\to \infty}
U\big(g_n\big).
\end{equation*}
Therefore, since by Assumption~\ref{ass:U-I} the sequence $\max\big\{U(g_n),0\big\}$, $n \in \N$, is uniformly integrable with respect to any $\P \in \cP$,  
Fatou's lemma for the medial limit (see \cite[Lemma~3.8(v)]{bartl2019limmed}) and \eqref{eq:pf_1-bounded} ensure that
\begin{align*}
\inf_{\P\in\mathcal{P}} \E_\P[U(\widehat{g})]
\geq 
\inf_{\P\in\mathcal{P}} \mathop{\mathrm{lim\,med}}_{n\to \infty} \E_\P\big[U(g_n)\big]
\geq \mathop{\mathrm{lim\,med}}_{n \to \infty} \big( \ov  u(x)- \tfrac{1}{n}\big)
= \ov  u(x).
\end{align*} 
This shows that indeed Item~\eqref{thm:res-3} holds and finishes the proof.
\end{proof}
%
%
It remains to prove Theorem~\ref{thm:main-2}.
To that end,
from now on, 
we denote for every $n \in \N$,
\begin{equation}\label{Un}
\begin{split}
&U_n(x):=U(x+\tfrac{1}{n}\big), \quad x \geq 0,\\
&V_n(y):= \sup_{x\geq 0} \big[U_n(x)-xy\big], \quad y\geq 0,
\end{split}
\end{equation}
and define $\ov u_n$ and $v_n$ as in \eqref{def:u-bar} and \eqref{def:v}, but with respect to $U_n$ and $V_n$, respectively.
Note that if $U$ satisfies Assumption~\ref{ass:U-2}, then each $U_n$, $n\in \N$, is a utility function which satisfies Assumption~\ref{ass:U-1}; in particular we can apply Theorem~\ref{thm:main} with respect to each $U_n$. This will be useful, by applying a limit argument, to prove Theorem~\ref{thm:main-2} for the case that $U$ satisfies Assumption~\ref{ass:U-2}.
%
%
\begin{lemma}
	\label{lem:conjugate.converges}
Let the assumptions in Theorem~\ref{thm:main-2} hold. 
Then for every $y>0$ we have that $\inf_{n \in \N} V_n(y)=V(y)$.
\end{lemma}
\begin{proof}
	Since $U_n\geq U$ it follows from the definition that $V_n\geq V$ for each $n \in \N$, and hence we focus on showing that $\inf_{n \in \N}V_n\leq V$.
	To that end,  fix some $y>0$ and let $(x_n)_{n \in \N} \subseteq [0,\infty)$ such that for each $n \in \N$
	\begin{equation}\label{eq:V-n-V-xn}
	V_n(y)= \sup_{x\geq 0} \big[ U(x+\tfrac{1}{n})-xy\big]
	\leq U\big(x_n+\tfrac{1}{n}\big)-x_ny + \tfrac{1}{n}.
	\end{equation}
	In particular, by monotonicity of $U$, we have that
	\begin{equation}\label{eq:V-n-V-1}
	\sup_{x\geq 0} \big[ U(x)-xy\big]
	\leq
	 \sup_{x\geq 0} \big[ U(x+\tfrac{1}{n})-xy\big]
	\leq U\big(x_n+\tfrac{1}{n}\big)-x_ny + \tfrac{1}{n}.
	\end{equation}
	Now notice that $U$ satisfying Assumption~\ref{ass:U-2} enforces that $\liminf_{n \to \infty} x_n>0$, since otherwise  
	$\liminf_{n\to \infty} U(x_n+1/n)-x_ny=U(0)=-\infty$, which contradicts \eqref{eq:V-n-V-1}. 
	Therefore, without loss of generality, we may assume that $x_n>0$ for each $n$.
	
	Moreover, we claim that $\limsup_{n \to \infty} x_n< \infty$. Indeed, if $\limsup_{n \to \infty} x_n=\infty$, then there is a subsequence (which we still denote by $(x_n)_{n \in \N}$)  such that 
	$\lim_{n \to \infty} x_n= \infty$. 
	Therefore,
	by concavity and monotonicity of $U$, we get that
	\begin{equation*}\label{eq:Vn-V-1}
	\tfrac{U(x_n)}{x_n}\leq	\tfrac{U(x_n+\frac{1}{n})}{x_n}\leq 	\Big(\tfrac{U(x_n)}{x_n}+ \tfrac{\partial_+ U(x_n)}{nx_n}\Big),
	\end{equation*}
	where $\partial_+ U$  denotes the right-derivative of $U$. Therefore, as $U$ is nondecreasing and concave satisfying Assumption~\ref{ass:U-2}, we obtain that
	\begin{equation*}
	\lim\limits_{n\to \infty} \tfrac{U(x_n+\frac{1}{n})}{x_n}=0.
	\end{equation*}
	For any fixed $0<\varepsilon<y$, we hence see  for big enough $n$ that
	\begin{equation*}
\big|\tfrac{U(x_n+\frac{1}{n})}{x_n} \big| \leq \varepsilon.
	\end{equation*}
	This ensures for any big enough $n$ that
	\begin{equation*}
	 U\big(x_n+\tfrac{1}{n}\big)-x_ny
	=
 x_n\Big( \tfrac{U(x_n+\frac{1}{n})}{x_n}-y \Big) \leq
 x_n\Big( \varepsilon-y \Big) <0.
	\end{equation*}
	This, in turn, implies that
	\begin{equation*}
	\lim_{n \to \infty} U\big(x_n+\tfrac{1}{n}\big)-x_ny
	=-\infty,
	\end{equation*}
	which contradicts \eqref{eq:V-n-V-1}.
	
	Therefore, we conclude that the sequence $(x_n)_{n \in \N}$ is bounded, and after passing to a subsequence, it has a limit $x\in(0,\infty)$.
	Thus by
	 \eqref{eq:V-n-V-xn} we obtain that
	\begin{equation*} 
	V(y)\geq U(x)-xy
	=\lim_{n \to \infty} \Big( U\big(x_n+\tfrac{1}{n}\big)-x_ny \Big)
	\geq\inf_{n \in \N} V_n(y),
	\end{equation*}
	which completes the proof.
\end{proof}
\noindent
Now we are ready to present the proof of our main results
 for the case where $U$ satisfies Assumption~\ref{ass:U-2}.
\begin{proof}
	[Proof of Theorem \ref{thm:main-2}]
	First, recall that Item~\eqref{thm:res-1-2} has been proved in Lemma~\ref{le:u-prop}.
	%
	
Furthermore, since each $U_n$ satisfies Assumption~\ref{ass:U-1}, we get from
Theorem~\ref{thm:main} 
that 
for every $n \in \N$  
\begin{equation}\label{eq:conj-un-vn}
\begin{split}
\ov u_n(x)&:=\inf_{y\geq 0} \big[v_n(y)+xy\big], \quad x > 0,\\
v_n(y)&:=\sup_{x\geq0} \big[\ov u_n(x)-xy\big], \quad y>0.
\end{split}
\end{equation}
Now, we claim that $\ov u(x)=\inf_n \ov u_n(x)$ for each $x>0$. Indeed, since by monotonicity $\ov u_n\geq \ov u$, we only need to show that $\ov u(x)\geq \inf_n \ov u_n(x)$.
To that end, fix $x> 0$. By 
 Theorem~\ref{thm:main}\eqref{thm:res-4} we have that $\ov u_n =u_n$, hence  there exists for each $n$ an element $g_n\in {\mathcal{C}}(x)$ such that 
\begin{equation}\label{eq:pf_1}
\ov u_n(x)
\leq 
\inf_{\P\in\mathcal{P}} \E_\P[U_n(g_n)]+\tfrac{1}{n}
= 
\inf_{\P\in\mathcal{P}} \E_\P\big[U(g_n+\tfrac{1}{n})\big]+\tfrac{1}{n}.
\end{equation}
	Define 
	\begin{equation*}
	\widehat{g}:=\mathop{\mathrm{lim\,med}}_{n \to \infty} g_n\in \ov{\mathcal{C}}(x).
	\end{equation*}
	Since $U$ is concave, we obtain by 
	 Jensen's inequality for medial limits (see \cite[Lemma~3.8(iii)]{bartl2019limmed})
	 that
	\begin{equation*}
	U(\widehat{g})
	=
	U\big(\mathop{\mathrm{lim\,med}}_{n\to \infty}(g_n+\tfrac{1}{n})\big)
	\geq 
	\mathop{\mathrm{lim\,med}}_{n\to \infty}
	U\big(g_n+\tfrac{1}{n}\big).
	\end{equation*}
Therefore, since by Assumption~\ref{ass:U-I} the sequence $\max\{U(g_n+1/n),0\}$, $n \in \N$, is uniformly integrable with respect to every $\P \in \cP$,  Fatou's lemma for the medial limit,  and \eqref{eq:pf_1} ensure that
	\begin{align*}
	\inf_{\P\in\mathcal{P}} \E_\P[U(\widehat{g})]
	\geq 
	\inf_{\P\in\mathcal{P}} \mathop{\mathrm{lim\,med}}_n \E_\P\big[U(g_n+\tfrac{1}{n})\big]
	\geq \mathop{\mathrm{lim\,med}}_n \big( \ov u_n(x)- \tfrac{1}{n}\big)
	=\inf_n \ov u_n(x).
	\end{align*} 
	This together with the fact that $\inf_n \ov u_n(x)\geq \ov u(x)$ shows that  for every $x>0$,
	\begin{equation}
\label{eq:un-u}
\inf_{\P\in\mathcal{P}} \E_\P[U(\widehat{g})]= \ov u(x)=\inf_n \ov u_n(x).
	\end{equation}
In particular, we see that Item~\eqref{thm:res-3-2} holds.

Next, we  claim that $\inf_n v_n(y)=v(y)$ for each $y>0$.
Indeed, by Lemma~\ref{lem:conjugate.converges} we know that $\inf_{n\in \N} V_n(y)=V(y)$ for every $y>0$, and since 
$n \mapsto V_n(y)$ is decreasing in $n$,
Assumption~\ref{ass:v-finite} together with the monotone convergence theorem  imply for every $y>0$ that
\begin{equation*}
\inf_n v_n(y)
=\inf_{\Q\in \mathcal{D},\,\P\in\mathcal{P}} \inf_n \E_\P\big[ V_n \big(y\tfrac{d\Q}{d\P}\big) \big]
=\inf_{\Q\in \mathcal{D},\,\P\in\mathcal{P}} \E_\P\big[ V\big(y\tfrac{d\Q}{d\P}\big) \big]
= v(y).
\end{equation*}
This, \eqref{eq:un-u}, and \eqref{eq:conj-un-vn} ensure that for every
$x>0$,
\begin{equation}\label{eq:u-rep} 
\begin{split}
\ov u(x)=\inf_{n \in \N} \ov u_n(x)
=\inf_{n \in \N}\inf_{y\geq 0} \big[v_n(y)+xy\big]
=\inf_{y\geq 0} \big[\inf_{n \in \N} v_n(y)+xy\big]
=\inf_{y\geq 0} [v(y)+xy\big].
\end{split}
\end{equation}

Furthermore, since by \eqref{eq:conj-un-vn} we know that each $v_n$ is nonincreasing, and as $\inf_n v_n(y)=v(y)$, we see that also $v$ is nonincreasing on $[0,\infty)$. 
In  addition, as $n \mapsto v_n(y)$ is nonincreasing for each $y>0$, and as each $v_n$ is convex, we conclude that also $v=\lim_n v_n$ is convex.
Moreover, by \eqref{eq:u-rep} we have for every $x>0$, $y\geq 0$ that  $\ov u(x)\leq v(y)+xy$, which together with Lemma~\ref{le:u-prop} imply that $v(y)>-\infty$ for all $y\geq 0$. In addition, by Assumption~\ref{ass:v-finite}, we get that
 $v(y)<\infty$ for all $y>0$. Therefore, we conclude that $v(y) \in \R$ for every $y>0$ and hence proves Item~\eqref{thm:res-1b-2}. Finally,  we can apply Lemma~\ref{le:V-extension}  together with \eqref{eq:u-rep} to conclude that 
 indeed  for every $y>0$ we have that
\begin{equation*}
v(y)=\sup_{x\geq 0} \big[\ov u(x)-xy\big],
\end{equation*}
which together with \eqref{eq:u-rep}  proves Item~\eqref{thm:res-2-2} and finishes the proof.
\end{proof}
%
%
%
\section{Proof of Theorem~\ref{thm:main-app}~\&~\ref{thm:main-app-2} and Corollary~\ref{co:main-app}~\&~\ref{co:main-app-2}}\label{sec:proof-app}
The idea of the proof of
Theorem~\ref{thm:main-app}~\&~\ref{thm:main-app-2} and Corollary~\ref{co:main-app}~\&~\ref{co:main-app-2}
is to verify that the assumptions in Theorem~\ref{thm:main} are satisfied. To that end, throughout this section, we put ourselves into the setting of Section~\ref{sec:application} and refer by $\cC$, $\cD$, $\cP$ to the corresponding sets specified there.

We recall the set of probability measures
\begin{equation*}
\fP_e(\cP):=\Big\{\Q \in \fP(\Omega)\colon  \exists\, \P \in \cP \mbox{ such that } \Q \approx \P
\Big\}
\end{equation*}
and  consider the following sets of probability measures, which will be useful in the sequel:
\begin{equation*}
\begin{split}
\fM_e(\cP)&:=\Big\{\Q \in \fP_e(\cP)\colon 
S \mbox{ is a  $\Q$-$\mathbb{F}$-local martingale}\Big\},\\
\cM&:=\Big\{\Q \in \cP^{ac}_{sem}\colon S \mbox{ is a $\Q$-$\F$-local martingale with } c^\Q \in \Theta_c \ \,\Q\otimes dt\mbox{-a.e.}\Big\}.
\end{split}
\end{equation*}
\begin{remark}\label{rem:M-P-tilde-Theta}
	By definition, we have that $\cM = \cP^{ac}_{sem}(\widetilde{\Theta})$ for the set $\widetilde\Theta:=\{0,...,0\}\times \Theta_c\subseteq \R^d \times \S^d_+$. In addition, due to Assumption~\ref{ass:invertible}, we will show in Proposition~\ref{prop:D-C-Cb-polar}
	that in fact $\cM=\fM_e(\cP)$.
\end{remark}
\subsection{Proof of Theorem~\ref{thm:main-app} and  Theorem~\ref{thm:main-app-2}}\label{subsec:proof-app}
%
%
\begin{lemma}\label{le:Q-P}
Let Assumption~\ref{ass:invertible} hold. Then for  each $\P \in \cP$ there exists $\Q\in \cM$ such that $\Q \approx \P$. 
Conversely, for each $\Q\in \cM$  there exists $\P \in \cP$  such that $\P \approx \Q$.
\end{lemma}
\begin{proof}
Let $\P \in \cP$ and consider the canonical decomposition of $S$ under $\P$ 
\begin{equation*}
S_t = \int_0^t b_s^\P \, ds + M^\P_t, \qquad 0\leq t \leq T,
\end{equation*}
where $M^\P$ is a continuous $\P$-local martingale with $\frac{d\langle M^\P \rangle}{dt}=c^\P$. Then Assumption~\ref{ass:invertible} guarantees that the stochastic process
\begin{equation}\label{eq:density-dQ-dP}
Z_t:= \mathcal{E}\Big(\int_0^t -\big[(c^\P)^{-1}\big]_s b^\P_s\,dM^\P_s\Big), \qquad 0\leq t \leq T,
\end{equation}
where $\cE(\cdot)$ denotes the stochastic exponential,
is well-defined and, e.g.\ by applying Novikov's condition, one sees that $Z$ defines a strictly positive continuous $\P$-martingale.
Therefore, one can define a measure $\Q \approx \P$ using  $(Z_t)_{0\leq t\leq T}$ as density process. Moreover, Girsanov's transformation theorem and Remark~\ref{rem:M-P-tilde-Theta} ensures that $\Q \in \cM$.

Conversely, let $\Q \in \cM$.
By \cite[Theorem~2.6]{neufeld2014measurability}, there exists an $\F$-predictable process such that $c=c^\Q \ \, \Q\otimes dt$-a.s.
Consider the set
\begin{equation}
\Upsilon:=\big\{(\omega,t)\in \Omega\times [0,T]\colon \exists\, b \in \R^d \mbox{ with } (b,c_t(\omega))\in \Theta \big\}.
\end{equation}
Since by assumption $\Theta\subseteq \R^d\times \S^d_+$ is compact (and hence closed), and the map $(\Omega\times [0,T])\times \R^d \ni (\omega,t,b)\mapsto (b,c_t(\omega))\in \R^d\times \S^d_+$ is a Carath\'eodory function, the implicit measurable functions theorem (see \cite[Theorem~14.16, p.654]{RockafellarWets}) ensures that $\Upsilon\subseteq \Omega\times[0,T]$ is an element of the $\F$-predictable $\sigma$-field
and that there exists a $\F$-predictable $\R^d$-valued stochastic process $(b_t)_{t\in [0,T]}$ such that
\begin{equation*}
(b_t(\omega),c_t(\omega))\in \Theta \quad \mbox{ for all } (\omega,t)\in \Upsilon.
\end{equation*}
Note that as $c=c^\Q \ \,\Q\otimes dt$-a.s., $\Theta_c=\proj_c(\Theta)$, and $\Q \in \cM$ we have that $\Upsilon$ has $\Q\otimes dt$-full measure.
Next, similar to above, 
Assumption~\ref{ass:invertible} guarantees that the  process
\begin{equation}\label{eq:density-dP-dQ}
\widetilde Z_t:= \mathcal{E}\Big(\int_0^t \big[c^{-1}\big]_s b_s\,dS_s\Big)
\end{equation}
is well-defined and, e.g.\ by applying Novikov's condition, one sees that $\widetilde Z$ defines a strictly positive continuous $\Q$-martingale.
Hence one can define a measure $\P \approx \Q$ using the  process $(\widetilde Z_t)_{t \in [0,T]}$ as density process. Moreover, Girsanov's transformation theorem ensures that $M^\P:= S-\int_0^\cdot b_s\,ds$ is a $\P$-local martingale. This in turn shows that
\begin{equation*}
\begin{split}
S_t= S_t-\int_0^t b_s\,ds + \int_0^t b_s\,ds = M^\P_t +\int_0^t b_s\,ds, \quad t\in [0,T],
\end{split}
\end{equation*}
which implies that $\P \in \cP$.
\end{proof}
As a consequence of the above lemma, we obtain the following observation.
\begin{remark}\label{rem:Nullsets}
Let Assumption~\ref{ass:invertible} hold. Then Lemma~\ref{le:Q-P} implies that the collection $\cN^{\cP}$ of all sets which are $\cF_T$-$\P$ null for every $\P \in \cP$ coincide with the corresponding set $\cN^{\cM}$. In particular, we see that
\begin{equation*}
\mathcal{G}_t = \bigcap_{s>t}  \big(\cF^*_s \vee \mathcal{N}^\cM\big), 
\qquad 0 \leq t \leq T.
\end{equation*}
\end{remark}
%
%
\begin{lemma}\label{le:diagonal}
Let Assumption~\ref{ass:invertible} holds. Then there exists  $\underline{\widehat{c}} \in \S^d_{++}$ which is diagonal and satisfies that $\underline{\widehat{c}}\leq c$ for all $c \in \Theta_c$.
\end{lemma}
\begin{proof}
Due to Assumption~\ref{ass:invertible}, there exists $\underline{c} \in \S^d_{++}$ which  satisfies that $\underline{c}\leq c$ for all $c \in \Theta_c$. Let  $\lambda_{\min}(\underline{c})>0$ 
be the smallest eigenvalue  
 of $\underline{c}$. 
Then we define
$\underline{\widehat{c}}=(\underline{\widehat{c}}^{ij})_{i,j\in\{1,\dots,d\}}$ by 
\begin{equation*}
\underline{\widehat{c}}^{ij}:= \lambda_{\min}(\underline{c}) \mbox{ Id}_{\R^{d\times d}}=
\begin{cases} 
\lambda_{\min}(\underline{c}) & \mbox{ if } i=j\\
0 & \mbox{ if } i\neq j.
\end{cases}
\end{equation*}
To see that $\underline{\widehat{c}}$ satisfies the desired properties, observe that $\underline{\widehat{c}}\in \S^d_{++}$  and is diagonal. Moreover, any eigenvalue of $\underline{c}-\underline{\widehat{c}}$ is of the form $\lambda_i-\lambda_{\min}(\underline{c})$ for some eigenvalue $\lambda_i$ of $\underline{c}$. This implies that
$\lambda_{\min}(\underline{c}-\underline{\widehat{c}})=0$, which ensures that 
$\underline{\widehat{c}}\leq  \underline{c}$.
\end{proof}
\begin{lemma}\label{le:H-Q-P-integr}
Let Assumption~\ref{ass:invertible} hold and let $(H_t)_{t\in [0,T]}$ be a $\mathbb{G}$-predictable process. Then $(H_t)_{t\in [0,T]}$ is $S$-integrable with respect to 
$\P$
for all $\P \in \cP$ if and only if   
$(H_t)_{t\in [0,T]}$ is $S$-integrable with respect to 
$\Q$
for all $\Q \in \cM$.
\end{lemma}
\begin{proof}
For the first direction, assume that $(H_t)_{t\in [0,T]}$ is $S$-integrable with respect to 
every
$\P \in \cP$, and let $\Q \in \cM$. By Lemma~\ref{le:Q-P} there exists $\P \in \cP$ such that  $\P\approx\Q$. Let 
\begin{equation*}
S=S_0+ M^\P+ \int_0^\cdot b^\P_s \, ds
\end{equation*}
be its canonical representation under $\P$, where $M^\P$ is a $\P$-local martingale with second differential characteristic $c^\P=(c^{ij,\P})_{i,j\in\{1,\dots,d\}}$. 
Then, as $H=(H^{(1)},\dots,H^{(d)})$ is  $S$-integrable with respect to 
$\P$
and  $\P\approx\Q$, we have that $c^\Q=c^\P \ \Q\otimes dt$-a.s.\ and  by \cite[Definition~III.6.17, p.207]{JacodShiryaev.03} that $\Q$-a.s.\ (and  $\P$-a.s.) 
\begin{equation*}
\int_0^T {\textstyle\sum\limits_{i,j=1}^d} H^{(i)}_s c^{ij,\Q}_s H^{(j)}_s\,ds 
=
\int_0^T {\textstyle\sum\limits_{i,j=1}^d} H^{(i)}_s c^{ij,\P}_s H^{(j)}_s\,ds
<\infty.
\end{equation*}
This implies by \cite[Theorem~III.6.4, p.204]{JacodShiryaev.03} that 
 $H$ is  $S$-integrable with respect to 
$\Q$.

On the other hand, assume now that $(H_t)_{t\in [0,T]}$ is $S$-integrable with respect every $\Q \in \cM$, and let $\P \in \cP$. 
Moreover, let 
\begin{equation*}
S=S_0+ M^\P+ \int_0^\cdot b^\P_s\,ds
\end{equation*}
be the canonical decomposition of $S$ under $\P$, 
 where $M^\P$ is a $\P$-local martingale with second differential characteristic $c^\P=(c^{ij,\P})_{i,j\in\{1,\dots,d\}}$.
By Lemma~\ref{le:Q-P} there exists $\Q \in \cM$ such that  $\Q\approx\P$.
Moreover, due to Assumption~\ref{ass:invertible}, we know from Lemma~\ref{le:diagonal} that there exists  $\underline{\widehat{c}} \in \S^d_{++}$ which is diagonal and satisfies that $\underline{\widehat{c}}\leq c$ for all $c \in \Theta_c$.
Therefore, by \cite[Theorem~III.6.4, p.204]{JacodShiryaev.03}, we have $\Q$-a.s.\ that
\begin{equation*}
\begin{split}
\sum_{i=1}^d \Big(\underline{\widehat{c}}^{ii} \int_0^T  |H^{(i)}_s|^2 \,ds\Big)
=
\int_0^T {\textstyle\sum\limits_{i,j=1}^d} H^{(i)}_s \underline{\widehat{c}}^{ij} H^{(j)}_s\,ds 
\leq 
\int_0^T {\textstyle\sum\limits_{i,j=1}^d} H^{(i)}_s c^{ij,\Q}_s H^{(j)}_s\,ds 
<\infty.
\end{split}
\end{equation*}
This and the fact that $\underline{\widehat{c}}^{ii}>0$ for each $i$ implies that each summand on the left-hand side is nonnegative and hence finite $\Q$-a.s.. In particular, we have for each $i\in \{1,\dots,d\}$ that $\Q$-a.s.\ (and hence also $\P$-a.s.),
\begin{equation}\label{eq:H-2-integrable}
\int_0^T  \big|H^{(i)}_s\big|^2 \,ds <\infty.
\end{equation}
Moreover,  the hypothesis that $\Theta$ is compact (and hence bounded) ensures that
$\cK:=\sup_{(b,c)\in \Theta} \big[ \Vert b\Vert + \Vert c\Vert\big]<\infty$.
This, \eqref{eq:H-2-integrable}, the fact that  $c^\P=c^\Q \ \P\otimes dt$-a.s., the Cauchy-Schwarz inequality, and \cite[Theorem~III.6.4, p.204]{JacodShiryaev.03} imply
 that $\P$-a.s. (and also $\Q$-a.s.)
\begin{equation*} 
\begin{split}
\int_0^T \Big| {\textstyle{\sum\limits_{i=1}^d}}H_s^{(i)}\ b^{i,\P}_s\Big|\,ds + \int_0^T {\textstyle{\sum\limits_{i=1}^d}}H^{(i)}_s c^{ij,\P}_s H^{(j)}_s\,ds
& \leq
\int_0^T  \cK{\textstyle{\sqrt{\sum\limits_{i=1}^d\big|H_s^{(i)}\big|^2}}}\ \,ds + \int_0^T {\textstyle{\sum\limits_{i=1}^d}}H^{(i)}_s c^{ij,\Q}_s H^{(j)}_s\,ds\\
& \leq
\cK\int_0^T  \Big[{1+\textstyle{\sum\limits_{i=1}^d\big|H_s^{(i)}\big|^2}}\Big]\ \,ds + \int_0^T {\textstyle{\sum\limits_{i=1}^d}}H^{(i)}_s c^{ij,\Q}_s H^{(j)}_s\,ds\\
&<\infty.
\end{split}
\end{equation*}
By  \cite[Definition~III.6.17, p.207]{JacodShiryaev.03}, we hence get that  $(H_t)_{t\in [0,T]}$ is  $S$-integrable with respect to $\P$, which finishes the proof.
\end{proof}
%
%
%
%
%
%
The following lemma is one of the two main tools to verify that the bipolar relation of $\cC$ and $\cD$ imposed in Theorem~\ref{thm:main} and Theorem~\ref{thm:main-2} hold. It states that on $\Omega=C([0,T],\R^d)$ the set of separating measures, which coincides with the set of local-martingale measures as $S$ has continuous sample paths, is already characterised by the separation of continuous functions. 
\begin{lemma}\label{le:separation}
	Let
\begin{equation}\label{def:Gamma}
\Gamma:=\Big\{ \gamma \in C_b(\Omega)\colon \mbox{ there exists }  H \in \mathcal{H}\mbox{ such that } \gamma \leq (H \cdot S)_T\Big\},
\end{equation}
and let $\Q \in \fP(\Omega)$ such that $\E_{\Q}[\gamma]\leq 0$ for all $\gamma \in \Gamma$. Then $(S_t)_{0\leq t \leq T}$ is a $\Q$-$\F$-local martingale.
\end{lemma}
\begin{proof}
This follows directly from Proposition~\ref{prop:separation-app} and Remark~\ref{rem:Ansel-Stricker-simple}, which in turn is a slight modification of \cite[Proposition~5.5]{predictionSet} and \cite[Proposition~4.4]{bartl2019duality}.
\end{proof}
The following two lemmas are direct consequences of Lemma~\ref{le:separation}.
\begin{proposition}\label{prop:D-C-Cb-polar}
We have that  
	\begin{equation*}
	\fM_e(\cP)=\cD= 
\big\{\Q \in \fP_e(\cP)\colon \E_\Q[X]\leq 1 \mbox{ for all } X \in (\mathcal{C}\cap C_b)\big\}.
\end{equation*}
In addition, if Assumption~\ref{ass:invertible} holds, then we additionally have that
	\begin{equation*}
\mathcal{M}=
\fM_e(\cP).
\end{equation*}
\end{proposition}
\begin{proof}
	Throughout this proof, let  $(\cC \cap C_b)^\circ:=\big\{\Q \in \fP_e(\cP)\colon \E_\Q[X]\leq 1 \mbox{ for all } X \in (\mathcal{C}\cap C_b)\big\}$.
	 
Now, to see that $\fM_e(\cP)\subseteq\mathcal{D}$, let $\Q \in \fM_e(\cP)$ and let $X \in \cC$. 
Then, there exists $\P \in \cP$ such that $\P\approx \Q$. This implies that
there exists $H \in \cH$ such that
$X\leq 1 + (H\cdot S)_T \ \Q$-a.s..
Therefore, since
$(H\cdot S)$ 
is  
a $\Q$-supermartingale by definition of the set $\cH$, we conclude that $\E_\Q[X]\leq 1$.  
Furthermore,
$\mathcal{D}\subseteq (\cC \cap C_b)^\circ$ follows directly from the definition of $(\cC \cap C_b)^\circ$. 

To see that $\fM_e(\cP)\supseteq (\cC \cap C_b)^\circ$, let $\Q \in (\cC \cap C_b)^\circ$. By definition, $\Q \in \fP_e(\cP)$. 
Now,  for each $\gamma \in \Gamma \subseteq C_b(\Omega)$
there exists $c\geq 0$ such that $c+\gamma\geq 0$. This 
implies that $\frac{1}{c}(c + \gamma) \in \cC \cap C_b$. This in turn implies that $\E_\Q[\frac{1}{c}(c + \gamma)]\leq 1$ which is equivalent to $\E_\Q[\gamma]\leq 0$. 
By 
Lemma~\ref{le:separation}
we get that $\Q$ is a local-martingale measure for $S$. This and the fact that $\Q \in \fP_e(\cP)$ implies that $\Q \in \fM_e(\cP)$.
Hence we have indeed shown that
\begin{equation*}
\fM_e(\cP)=\cD= (\cC \cap C_b)^\circ.
\end{equation*}
Finally, 
if Assumption~\ref{ass:invertible} holds, 
then $\mathcal{M}\subseteq\fM_e(\cP)$ follows directly from Lemma~\ref{le:Q-P}. Conversely,  $\fM_e(\cP)\subseteq \mathcal{M}$ follows by Girsanov's theorem for semimartingales \cite[Theorem~III.3.24, p.172]{JacodShiryaev.03} and the fact that a semimartingale with continuous sample paths is a local martingale if and only if its predictable finite-variation part vanishes.
\end{proof}	
\begin{proposition}\label{prop:P-D-convex-compact}
Let Assumption~\ref{ass:invertible} hold. Then both $\cP, \cD \subseteq \fP(\Omega)$ are convex and compact.
\end{proposition}
\begin{proof}
First, note that by definition $\cM=\cP^{ac}_{sem}(\widetilde{\Theta})$, where $\widetilde{\Theta}:=\{0,...,0\}\times\Theta_c\subseteq \R^d\times \S^d_+$ and $\Theta_c:=\proj_c(\Theta)\subseteq\S^d_+$. Moreover, as by assumption $\Theta$ is convex and compact, so is $\widetilde{\Theta}$. Therefore, the compactness of $\mathcal P$ and $\cM$ follows directly from \cite[Theorem~2.5]{liu2019compactness}, whereas the convexity of $\cP$ and $\cM$ follows by \cite[Theorem~III.3.40, p.176]{JacodShiryaev.03}.
In addition, we know by Proposition~\ref{prop:D-C-Cb-polar}  that $\cM=\mathcal{D}$, which finishes the proof.
\end{proof}
The following lemma is the second crucial tool to prove the bipolar relation imposed on $\cC$ and $\cD$. It heavily uses the fact that one can construct a process $Y$ which is a $\Q$-supermartingale for every $\Q\in \cM$, as well as the robust optional decomposition theorem.
\begin{proposition}\label{prop:C-C-b-char}
Let Assumption~\ref{ass:invertible} hold. Then we have that 
\begin{equation*}
\big\{ X \in C_b^+\colon \E_{\mathbb Q}[X]\leq 1 \mbox{ for all }\mathbb Q \in \mathcal{D} \big\}=\mathcal{C}\cap C_b.
\end{equation*}
\end{proposition}
\begin{proof}
	By definition, we have that $\big\{ X \in C_b^+\colon \E_{\mathbb Q}[X]\leq 1 \mbox{ for all }\mathbb Q \in \mathcal{D} \big\}\supseteq\mathcal{C}\cap C_b$. To see the opposite inclusion, let $X \in C_b^+$ such that $\E_{\mathbb Q}[X]\leq 1$ for all $\mathbb Q \in \mathcal{D}$. Since $X$ is nonnegative, bounded, and continuous (and so Borel), and since by \cite[Theorem~2.1]{neufeld2017nonlinear} the set $\cM$ satisfies the so-called Condition~(A) (see \cite{neufeld2017nonlinear} or \cite{nutz2015robust} for the precise definition), we can apply the same argument as in the proof of  \cite[Theorem~3.2]{nutz2015robust} and  \cite[Theorem~2.3]{NeufeldNutz.12}
	 and use Remark~\ref{rem:Nullsets} to obtain a  $\mathbb G$-adapted nonnegative process $(Y_t)_{0\leq t \leq T}$ with c\`adl\`ag sample paths which is a $\mathbb{Q}$-$\mathbb G$-supermartingale for every $\mathbb Q \in \mathcal M$ and satisfies that
	\begin{equation}\label{eq:RobOptDec-1}
	Y_0\leq \sup_{\mathbb{Q} \in \mathcal{M}}\E_{\Q}[X] \qquad \mbox{ as well as }  \qquad Y_T=X \qquad \Q\mbox{-a.s.\quad for all } \quad  \Q \in \mathcal{M}.
	\end{equation}
	Moreover, since the set $\cM$ is saturated (in the notion of \cite{nutz2015robust}, see also \cite[Lemma~4.2]{nutz2015robust}), the robust optional decomposition theorem (see \cite[Theorem~2.4]{nutz2015robust})
	ensures the existence of a $\mathbb G$-predictable process $H$ such that $H$ is $S$-integrable for all $\mathbb Q \in \mathcal M$ 
	and
	\begin{equation}\label{eq:RobOptDec-2}
	Y- (H \cdot S) \mbox{ is nonincreasing } \qquad \mathbb Q\mbox{-a.s. } \quad \mbox{ for all } \quad \Q \in \mathcal{M}.
	\end{equation}
	Combining this, \eqref{eq:RobOptDec-1}, and the fact that by Proposition~\ref{prop:D-C-Cb-polar} we know that $\mathcal{M}=\mathcal{D}$ implies that
	\begin{equation}\label{eq:RobOptDec-3}
1 + (H \cdot S)_T \geq Y_0 + (H \cdot S)_T \geq Y_T = X \quad \mathbb Q\mbox{-a.s. } \quad \mbox{ for all } \qquad \Q \in \mathcal{M}.
	\end{equation}
	Moreover, for any $\Q \in \cM$
	we use \eqref{eq:RobOptDec-1}, \eqref{eq:RobOptDec-2}, that $Y\geq 0$ is  a $\Q$-supermartingale, and that $\cM=\cD$ to see that 
	\begin{equation}
	(H \cdot S)_t\geq Y_t-Y_0\geq \E_{\Q}[X\,|\,\cG_t]-1\geq -1 \qquad \Q\mbox{-a.s. }
	\end{equation}
 for all $t\in [0,T]$. 
 Therefore, we conclude that  $(H \cdot S)\geq -1 \ \, \cM$-q.s., which by Lemma~\ref{le:Q-P} implies that $(H \cdot S)\geq -1 \ \, \cP$-q.s..
This	and Lemma~\ref{le:H-Q-P-integr}  ensure that $H \in \mathcal{H}$. Moreover, Lemma~\ref{le:Q-P}  and \eqref{eq:RobOptDec-3} assure that
		\begin{equation*}
	1 + (H \cdot S)_T \geq X \quad \mathbb P\mbox{-a.s. } \quad \mbox{ for all } \quad \P \in \mathcal{P},
	\end{equation*}
	which by definition shows that $X \in \cC$. As by assumption $X\in C_b^+$ we indeed get that $X \in \cC\cap  C_b$.
\end{proof}
%
%
Now we are able to finish the proof of Theorem~\ref{thm:main-app}.
\begin{proof}[Proof of Theorem~\ref{thm:main-app} and Theorem~\ref{thm:main-app-2}]
	We verify that Assumption~\ref{ass:no-arbitrage-style}, the bipolar relation of $\cC$ and $\cD$, and the convex-compactness assumption on $\cP$ and $\cD$ are satisfied. 
	
	To that end, note that Lemma~\ref{le:Q-P} ensures that Assumption~\ref{ass:no-arbitrage-style} hold.
	Moreover, observe that Proposition~\ref{prop:P-D-convex-compact} assures that $\cP$ and $\cD$ (with $\fP:=\fP_e(\cP)$) are both convex and compact (compare with Item~\eqref{thm:ass-1}). 
	In addition, we get from Proposition~\ref{prop:D-C-Cb-polar} and Proposition~\ref{prop:C-C-b-char} that the bipolar relation of $\cC$ and $\cD$ (compare with Items~\eqref{thm:ass-2}~\&~\eqref{thm:ass-3}) hold. Therefore, the result now follows directly from Theorem~\ref{thm:main} and Theorem~\ref{thm:main-2}, respectively.
\end{proof}
%
%
\subsection{Proof of  Corollary~\ref{co:main-app} and Corollary~\ref{co:main-app-2}}\label{subsec:proof-co}
The idea of the proof of Corollary~\ref{co:main-app} and Corollary~\ref{co:main-app-2}  is to verify in the setting of drift and volatility uncertainty introduced in Section~\ref{sec:application} that Assumptions~\ref{ass:u-finite}~\&~\ref{ass:u-bar-finite} and Assumptions~\ref{ass:U-I}~\&~\ref{ass:v-finite} hold for the specific utility functions. Then the result immediately follows from Theorem~\ref{thm:main-app} and Theorem~\ref{thm:main-app-2}.

First, note that 
every utility function  which is bounded from above automatically satisfies Assumptions~\ref{ass:u-finite}~\&~\ref{ass:u-bar-finite} as well as Assumptions~\ref{ass:U-I}~\&~\ref{ass:v-finite}; see also Remark~\ref{rem:ass:U-I} and Remark~\ref{rem:ass:U-I-2}. Therefore, we only have to focus on the utility functions $U(x):=\log(x)$ and $U(x):=\tfrac{x^p}{p}$, $p\in(0,1)$, to prove Corollary~\ref{co:main-app}
and Corollary~\ref{co:main-app-2}.

Moreover, observe that due to Assumption~\ref{ass:invertible}, we have that
\begin{equation}
\label{def:kappa}
\cK:=
 1 + \sup_{(b,c)\in \Theta} \big[\Vert b \Vert+ \Vert c \Vert + \Vert c^{-1} \Vert\big]
 <\infty.
\end{equation}
The following lemma will be used several times in this subsection.
\begin{lemma}\label{le:density-estimate}
Let Assumption~\ref{ass:invertible} hold. Then for every $\P \in  \cP$ there exists $\Q\in \fM_e(\cP)$ such that for every $\delta \in (0,\infty)$
\begin{equation}
\E_{\Q} \Big[\big(\tfrac{d\P}{d\Q}\big)^\delta\Big]<\infty.
\end{equation}
\end{lemma}
\begin{proof}
Let $\P \in \cP$. Jensen's inequality, Lemma~\ref{le:Q-P}, and Proposition~\ref{prop:D-C-Cb-polar} ensure that the statement holds for $\delta\in(0,1]$, hence we only need to focus on the case $\delta>1$. Note that from the proof of Lemma~\ref{le:Q-P}, see \eqref{eq:density-dQ-dP} and \eqref{eq:density-dP-dQ}, we know that
there
 exists $\Q\in \fM_e(\cP)$ such that
\begin{equation}\label{eq:dP-dQ-proof}
\tfrac{d\P}{d\Q}= \cE\Big(\int_0^T \big[(c^\P)^{-1}\big]_s b^\P_s\,dS\Big) \qquad \Q\mbox{-a.s.,}
\end{equation}
where $(b^\P,c^\P)$ denotes the differential characteristics of $S$ under $\P$. This and the fact that $c^\Q=c^\P \ \Q\otimes dt$-a.s. imply for every $\delta >1$ that
\begin{equation*}
\begin{split}
&\E_{\Q}\bigg[\cE\Big(\int_0^T \big[(c^\P)^{-1}\big]_s b^\P_s\,dS\Big)^\delta\bigg]\\
&=
\E_{\Q}\bigg[\exp\Big(\int_0^T \big[(c^\P)^{-1}\big]_s b^\P_s\,dS-\tfrac{1}{2} \int_0^T {\textstyle\sum\limits_{i,j=1}^d}b^{i,\P}_s\big[(c^\P)^{-1}\big]_s^{ij}b^{j,\P}_s\,ds\Big)^\delta\bigg]\\
&=
\E_{\Q}\Bigg[\exp\Big(\int_0^T \delta \big[(c^\P)^{-1}\big]_s b^\P_s\,dS-\tfrac{1}{2} \int_0^T {\textstyle\sum\limits_{i,j=1}^d}\delta^2 b^{i,\P}_s\big[(c^\P_s)^{-1}\big]^{ij}b^{j,\P}_s\,ds\Big)\\
& \quad \quad \quad \ \cdot
\exp\Big(\tfrac{1}{2}(\delta^2-\delta) \int_0^T {\textstyle\sum\limits_{i,j=1}^d} b^{i,\P}_s\big[(c^\P)^{-1}\big]_s^{ij}b^{j,\P}_s\,ds\Big)
\Bigg]\\
&=\E_{\Q}\bigg[\cE\Big(\int_0^T \delta\big[(c^\P)^{-1}\big]_s b^\P_s\,dS\Big)\, \exp\Big(\tfrac{1}{2}(\delta^2-\delta) \int_0^T {\textstyle\sum\limits_{i,j=1}^d} b^{i,\P}_s\big[(c^\P)^{-1}\big]_s^{ij}b^{j,\P}_s\,ds\Big)\bigg].
\end{split}
\end{equation*}
This, the fact that by Assumption~\ref{ass:invertible} we have \eqref{def:kappa}, and the fact that $S$ under $\Q$ is a (local-) martingale  show that for every $\delta >1$, we indeed have that
\begin{equation*}
\begin{split}
\E_{\Q}\bigg[\cE\Big(\int_0^T \big[(c^\P)^{-1}\big]_s b^\P_s\,dS\Big)^\delta\bigg]
& \leq 
\E_{\Q}\bigg[\cE\Big(\int_0^T \delta\big[(c^\P)^{-1}\big]_s b^\P_s\,dS\Big)\bigg]
\exp\Big(\tfrac{1}{2}(\delta^2-\delta) T d^2\cK^3\Big)\\
&=  \exp\Big(\tfrac{1}{2}(\delta^2-\delta) T d^2\cK^3\Big)\\
&< \infty.
\end{split}
\end{equation*}
\end{proof}
\begin{lemma}\label{le:cons-density-estimate}
	Let 
	Assumption~\ref{ass:invertible} hold, and let $x>0$ and $(g_n)_{n \in \N} \subseteq  \cC(x)$. 
	Then for every $\P \in \cP$ and every $\varepsilon\in (0,1)$ we have that
	\begin{equation*}
	\sup_{n \in \N} \E_\P\big[(g_n)^\varepsilon\big]<\infty.
	\end{equation*}
	\end{lemma}
\begin{proof}
	Fix $\varepsilon\in (0,1)$, $n \in \N$, and $\P \in \cP$.
	By Lemma~\ref{le:density-estimate}, there exists $\Q \in \fM_e(\cP)$ which satisfies for every $\delta\in (0,\infty)$ that $c(\delta):=\E_{\Q} \big[(\tfrac{d\P}{d\Q})^\delta\big]<\infty$. Therefore, H\"olders inequality (applied to $p:=\tfrac{1}{1-\varepsilon}$, $q:=\tfrac{1}{\varepsilon}$) and the fact that $\cD=\fM_e(\cP)$ (see Proposition~\ref{prop:D-C-Cb-polar}) ensure 
	 that indeed
	\begin{equation*}
	\begin{split}
	\E_\P\big[(g_n)^\varepsilon\big] \leq  \E_\Q\big[\tfrac{d\P}{d\Q}(g_n)^\varepsilon\big] 
	\leq 
	\E_\Q\big[\big(\tfrac{d\P}{d\Q}\big)^{\nicefrac{1}{(1-\varepsilon)}}\big]^{(1-\varepsilon)}
	\E_\Q\big[g_n\big]^{\varepsilon}
	\leq c(\nicefrac{1}{(1-\varepsilon)})^{1-\varepsilon} x^\varepsilon <\infty.
	\end{split}
	\end{equation*}
\end{proof}
	
\begin{lemma}\label{le:log-U-I}
Let 
Assumption~\ref{ass:invertible} hold, and let $U(x):=\log(x)$. Then for every $x>0$, $(g_n)_{n \in \N} \subseteq \cC(x) $ we have that the sequence of random variables
\begin{equation*}
\max\big\{\log(g_n+ \tfrac{1}{n}),0\big\}, \ n \in \N,
\end{equation*}
is $\P$-uniformly integrable for every $\P \in \cP$. 
\end{lemma}
\begin{proof}
		Fix $\P \in \cP$,  let $\varepsilon \in (0,1)$, and define the function $\Psi\colon [0,\infty) \to [0,\infty)$ by $\Psi(x)=\exp(\varepsilon x)$.
	Then, by the de la Vall\'ee-Poussin theorem, it suffices to show 
	that
	\begin{equation*}
	\sup_{n \in \N} \E_\P\Big[\Psi\Big(\max\big\{\log(g_n+ \tfrac{1}{n}),0\big\}\Big)\Big]<\infty.
	\end{equation*}
	Since $x \mapsto \Psi(x)=\exp(\varepsilon x)$ is increasing,
	we have for every $n\in \N$ that 
	\begin{equation*}
	\Psi\Big(\max\big\{\log(g_n+ \tfrac{1}{n}),0\big\}\Big)
	=\max\Big\{\Psi\big(\log(g_n+ \tfrac{1}{n})\big),\Psi(0)\Big\}
	=\max\Big\{(g_n+ \tfrac{1}{n})^\varepsilon,1\Big\},
	\end{equation*}
	hence it suffices to show 
	 that
	\begin{equation*}
	\sup_{n \in \N} \E_\P\big[(g_n+ \tfrac{1}{n})^\varepsilon\big]<\infty.
	\end{equation*}
	Therefore, as $(g_n+ \tfrac{1}{n})^\varepsilon \leq (g_n)^\varepsilon  + (\tfrac{1}{n})^\varepsilon$ for each $n \in \N$, it suffices to show 
	that
	\begin{equation}\label{eq:log-UI}
	\sup_{n \in \N} \E_\P\big[(g_n)^\varepsilon\big]<\infty.
	\end{equation}
	Lemma~\ref{le:cons-density-estimate}
	hence indeed implies that \eqref{eq:log-UI} holds.
\end{proof}
%
%
\begin{lemma}\label{le:power-U-I}
	Let 
	Assumption~\ref{ass:invertible} hold, and let $U(x):=\frac{x^p}{p}$ for some $p\in (0,1)$. Then for every $x>0$, $(g_n)_{n \in \N} \subseteq  \cC(x) $ we have that the sequence of random variables
	\begin{equation*}
	\frac{(g_n+ \frac{1}{n})^p}{p}, \ n \in \N,
	\end{equation*}
	is $\P$-uniformly integrable for every $\P \in \cP$. 
\end{lemma}
\begin{proof}
		Fix $\P \in \cP$,  let $\varepsilon \in (p,1)$, and define the function $\Psi\colon [0,\infty) \to [0,\infty)$ by $\Psi(x)=x^{\nicefrac{\varepsilon}{p}}$.
Then, by the de la Vall\'ee-Poussin theorem, it suffices to show 
that
\begin{equation*}
\sup_{n \in \N} \E_\P\Big[\Psi\Big(\tfrac{(g_n+ \frac{1}{n})^p}{p}\Big)\Big]<\infty.
\end{equation*}
To see this, since $\Psi(x)=x^{\nicefrac{\varepsilon}{p}}$, it suffices to show 
that
\begin{equation*}\label{eq:power-UI}
\sup_{n \in \N} \E_\P\big[(g_n)^\varepsilon\big]<\infty,
\end{equation*}
which follows directly from Lemma~\ref{le:cons-density-estimate}.
\end{proof}
Now, recall that
\begin{equation*}
\begin{split}
V_1(y):=\sup_{x\geq 0}  \big[U_1(x)-xy\big], \quad y>0, 
\end{split}
\end{equation*}
where $U_1(x):=U(x + 1)$, $x\geq 0$. Then we have the following.
\begin{lemma}\label{le:conjugate-log-power}
For each $p\in (0,1)$, $y>0$ we have that 
\begin{equation}\label{eq:conjugate-log-power-1}
\begin{split}
V_{1,\log}(y)&:=\sup_{x\geq 0}\big[\log(x+1)-xy\big]\leq \log\big(\tfrac{1}{y})-1+y,\\
%
%
V_{1,p}(y)&:=\sup_{x\geq 0}\big[\tfrac{(x+1)^p}{p}-xy\big] 
\leq (\tfrac{1}{p}-1\big) \big(\tfrac{1}{y}\big)^{\frac{p}{1-p}}+y.
\end{split}
\end{equation}
\end{lemma}
\begin{proof}
We start for the $\log$-case.
To that end, for every $y>0$, let $\widehat{x}_{1,\log}(y):=\tfrac{1}{y}-1$.
Then one sees, using the first-order condition,
that for every $y>0$,
\begin{equation*}
V_{1,\log}(y)
\leq \sup_{x\geq -1}\big[\log(x+1)-xy\big]
=\big[\log\big(\widehat{x}_{1,\log}(y)+1\big)-\widehat{x}_{1,\log}(y)y\big]
=
\log\big(\tfrac{1}{y})-1+y.
\end{equation*}
To see the result for the power-case, we set for every $y>0$ that $\widehat{x}_{1,p}(y):=y^{\nicefrac{1}{(p-1)}}-1$. Then, using the first-order condition, we get for every $y>0$ that
\begin{equation*}
V_{1,p}(y)
\leq
 \sup_{x\geq -1}\big[\tfrac{(x+1)^p}{p}-xy\big] 
 = \Big[\tfrac{\big(\widehat{x}_{1,p}(y)+1\big)^p}{p}-\widehat{x}_{1,p}(y)y\Big]
 = (\tfrac{1}{p}-1\big) \big(\tfrac{1}{y}\big)^{\frac{p}{1-p}}+y.
\end{equation*}
\end{proof}
%
%
%
%
%
\begin{lemma}\label{le:log-V}
	Let Assumption~\ref{ass:invertible} hold. Then, for every $y>0$ and every $\P \in \cP$ there exists $\Q \in \cD$ such that
	\begin{equation*}
\E_\P\big[\!\max\!\big\{V_{1,\log}(y\tfrac{d\Q}{d\P}),0\big\}\big]
<\infty.
\end{equation*}
\end{lemma}
\begin{proof}
Let $y>0$ and $\P \in \cP$. By Proposition~\ref{prop:D-C-Cb-polar}, we know that $\cD=\fM_e(\cP)$. Moreover,
by Lemma~\ref{le:density-estimate}, there exists $\Q \in \fM_e(\cP)$ which satisfies for every $\delta\in (0,\infty)$ that $c(\delta):=\E_{\Q} \big[(\tfrac{d\P}{d\Q})^\delta\big]<\infty$.
This and Lemma~\ref{le:conjugate-log-power} imply that
\begin{equation*}
\begin{split}
\E_\P\big[\!\max\!\big\{V_{1,\log}(y\tfrac{d\Q}{d\P}),0\big\}\big]
&=
\E_\P\Big[\!\max\!\big\{\log\big(\tfrac{1}{y}\tfrac{d\P}{d\Q}\big)-1+y\tfrac{d\Q}{d\P},0\big\}\Big]
\\
& \leq \max\!\big\{\log\big(\tfrac{1}{y}),0\big\} 
+ \E_\P\Big[\!\max\!\big\{\log\big(\tfrac{d\P}{d\Q}\big),0\big\}\Big]
+ y\E_{\P}\big[\tfrac{d\Q}{d\P}\big].
\end{split}
\end{equation*}
Since $\E_{\P}\big[\tfrac{d\Q}{d\P}\big]=1$, 
it  hence suffices to show that $\E_\P\big[\!\max\!\big\{\log\big(\tfrac{d\P}{d\Q}\big),0\big\}\big]<\infty$. To see this, note that the fact that $\log(x)\leq 1 +x $ for all $x\geq 0$ and Lemma~\ref{le:density-estimate} indeed ensure that
\begin{equation*}
\begin{split}
\E_\P\Big[\!\max\!\big\{\log\big(\tfrac{d\P}{d\Q}\big),0\big\}\Big]
&= \E_\Q\Big[\tfrac{d\P}{d\Q} \max\!\big\{\log\big(\tfrac{d\P}{d\Q}\big),0\big\}\Big]
\leq 
\E_\Q\Big[\tfrac{d\P}{d\Q} \big(1+\tfrac{d\P}{d\Q}\big)\Big]
<\infty.
\end{split}
\end{equation*}
\end{proof}
\begin{lemma}\label{le:power-V}
	Let Assumption~\ref{ass:invertible} hold. Then, for every $y>0$ and every $\P \in \cP$ there exists $\Q \in \cD$ such that
	\begin{equation*}
	\E_\P\big[\!\max\!\big\{V_{1,p}(y\tfrac{d\Q}{d\P}),0\big\}\big]
	<\infty.
	\end{equation*}
\end{lemma}
\begin{proof}
Let $y>0$ and $\P \in \cP$. By Proposition~\ref{prop:D-C-Cb-polar}, we know that $\cD=\fM_e(\cP)$. Moreover,
by Lemma~\ref{le:density-estimate}, there exists $\Q \in \fM_e(\cP)$ which satisfies for every $\delta\in (0,\infty)$ that $c(\delta):=\E_{\Q} \big[(\tfrac{d\P}{d\Q})^\delta\big]<\infty$.
This, the fact that $\tfrac{d\Q}{d\P} \in \mathcal{L}^1(\P)$, and Lemma~\ref{le:conjugate-log-power} imply that
\begin{equation*}
\begin{split}
\E_\P\big[\!\max\!\big\{V_{1,p}(y\tfrac{d\Q}{d\P}),0\big\}\big]
&=
(\tfrac{1}{p}-1\big)\big(\tfrac{1}{y}\big)^{\frac{p}{1-p}}\,\E_\P\Big[\big(\tfrac{d\P}{d\Q}\big)^{\frac{p}{1-p}}\Big]
+y\E_{\P}\big[\tfrac{d\Q}{d\P}\big]\\
&=
(\tfrac{1}{p}-1\big)\big(\tfrac{1}{y}\big)^{\frac{p}{1-p}}\,\E_\Q\Big[\big(\tfrac{d\P}{d\Q}\big)^{1+\frac{p}{1-p}}\Big]
+ y\E_{\P}\big[\tfrac{d\Q}{d\P}\big]
<\infty.
\end{split}
\end{equation*}
\end{proof}

\begin{lemma}\label{le:u-finite}
	 Let the utility function $U$ be either $U(x):=\log(x)$ or $U(x):=\tfrac{x^p}{p}$ for some $p\in (0,1)$, and let Assumption~\ref{ass:invertible} hold.
	Then for every $x>0$ we have that $u(x)<\infty$.
	If in addition Assumption~\ref{ass:limmed} holds, then we also have that $\ov u(x)<\infty$.
\end{lemma}
\begin{proof}
By Lemma~\ref{le:density-estimate}, we know that for every $\P$ there exists $\Q_\P \in \fM_e(\cP)$ which satisfies for every $\delta\in (0,\infty)$ that $c(\delta):=\E_{\Q_\P} \big[(\tfrac{d\P}{d\Q_\P})^\delta\big]<\infty$.
By the weak duality, see also \eqref{eq:weak-duality-U-V}--\eqref{eq:weak-duality-U-V-2}, the fact that $V(y)\leq V_1(y)$ for every $y\geq 0$, and Lemma~\ref{le:log-V}~\&~\ref{le:power-V} we see  that indeed
for every $x>0$, $y>0$,
\begin{equation*}
u(x)\leq v(y)+ xy \leq \inf_{\P\in \cP}\E_\P\big[\!\max\!\big\{V_1(y\tfrac{d\Q_\P}{d\P}),0\big\}\big]+xy <\infty.
\end{equation*}
If Assumption~\ref{ass:limmed} holds in addition, then the same arguments, together with the weak duality with respect to $\ov u$ derived in \eqref{eq:weak-duality-U-V-1}, show that also $\ov u(x)<\infty$ for all $x>0$.
\end{proof}
We are now able to provide the proof of Corollary~\ref{co:main-app} and Corollary~\ref{co:main-app-2}.
\begin{proof}[Proof of Corollary~\ref{co:main-app} and Corollary~\ref{co:main-app-2}]
We verify that the
Assumptions~\ref{ass:u-finite}~\&~\ref{ass:u-bar-finite} and Assumptions~\ref{ass:U-I}~\&~\ref{ass:v-finite}  hold for the specific utility functions. 

First, note that the specific 
 utility functions $U\colon[0,\infty) \to [-\infty,\infty)$ defined by $U(x)=\tfrac{x^p}{p}$, $p\in(0,1)$, and  $U(x)=-e^{-\lambda x}$, $\lambda>0$, satisfy Assumption~\ref{ass:U-1}, whereas the utility functions
$U(x):=\log(x)$ and $U(x):=\tfrac{x^p}{p}$, $p\in (-\infty,0)$, satisfy Assumption~\ref{ass:U-2}. Moreover, note that every utility function  which is bounded from above automatically satisfies the Assumptions~\ref{ass:u-finite}~\&~\ref{ass:u-bar-finite} and Assumptions~\ref{ass:U-I}~\&~\ref{ass:v-finite}; see also Remark~\ref{rem:ass:U-I} and Remark~\ref{rem:ass:U-I-2}. 
Therefore, we only need to show that the Assumptions~\ref{ass:u-finite}~\&~\ref{ass:u-bar-finite} and Assumptions~\ref{ass:U-I}~\&~\ref{ass:v-finite}  hold for the utility functions $U(x):=\log(x)$ and $U(x)=\tfrac{x^p}{p}$, $p\in(0,1)$.

To that end, note that for these utility functions, Lemma~\ref{le:u-finite} guarantees that Assumptions~\ref{ass:u-finite}~\&~\ref{ass:u-bar-finite} hold,  Lemma~\ref{le:log-V}~\&~\ref{le:power-V} show that Assumption~\ref{ass:v-finite} hold, whereas Lemmas~\ref{le:log-U-I}~\&~\ref{le:power-U-I} ensure that Assumption~\ref{ass:U-I} hold. Therefore, Corollary~\ref{co:main-app} now directly follows from Theorem~\ref{thm:main-app} together with the fact that  Assumption~\ref{ass:v-finite} implies that $v(y)<\infty$ for every $y>0$, 
whereas Corollary~\ref{co:main-app-2} now directly follows from Theorem~\ref{thm:main-app-2}.
\end{proof}
\appendix
\section{Appendix: Continuous separation}\label{sec:appendix}
\setcounter{theorem}{0}
\setcounter{equation}{0}

%
%
Throughout this section, we  will work in the framework of Section~\ref{sec:application}.
We recall that $\Omega:=C([0,T],\R^d)$ is endowed with its Borel $\sigma$-field $\cF$. Moreover, we let  $(S_t)_{0\leq t \leq T}$ be the canonical process. In addition, we let $\F$ be the raw filtration generated the canonical process and $\F_+$ denotes the corresponding right-continuous version of $\F$.

We define the set  $\mathcal{H}_{s,d}(\F_+)$ of all $d$-dimensional ${\F}_+$-simple processes $H\colon\Omega\times[0,T]\to \R^d$ of the form 
$H_t(\omega):=\sum_{\ell=1}^L h_\ell(\omega)\mathbf{1}_{(\tau_{\ell},\tau_{\ell+1}]}(t)$ for $(\omega,t)\in \Omega\times [0,T]$, where $L \in \N$, 0 $\leq \tau_1\leq \dots \leq \tau_{L+1}\leq T$ are ${\mathbb{F}}_+$-stopping times, and
$h_\ell:=(h_\ell^{(1)},\dots,h_\ell^{(d)})\colon \Omega \to \R^d$ is bounded and $\cF_{\tau_\ell+}$-measurable. 
Furthermore we define for every $m \in \N$ the set
\begin{equation*}
\mathcal{H}_{s,d,m}(\F_+):=\Big\{H \in \mathcal{H}_{s,d}(\F_+) \colon (H\cdot S) \geq -m \mbox{ pointwise on } \Omega \times [0,T] \Big\}.
\end{equation*}
\begin{remark}\label{rem:Ansel-Stricker-simple}
	Recall  the filtration $\mathbb{G}$ and  the set of strategies $\mathcal{H}$ introduced in Section~\ref{sec:application}.
Then, the fact that ${\F}_+\subseteq \mathbb G$
 immediately implies that
$\mathcal{H}_{s,d,m}({\F}_+)\subseteq \mathcal{H}$ for each $m \in \N$. 
\end{remark}
The following result slightly extends \cite[Proposition~5.5]{predictionSet} and \cite[Proposition~4.4]{bartl2019duality}.
\begin{proposition}\label{prop:separation-app}
Consider the set
	\begin{equation}\label{def:Gamma-app}
	\Gamma_d:=\Big\{ \gamma \in C_b(\Omega)\colon \mbox{ there exists } m \in \N \mbox{ and }  H \in \mathcal{H}_{s,d,m}(\F_+) \mbox{ such that } \gamma \leq (H \cdot S)_T\Big\}.
	\end{equation}
	Moreover, let $\Q \in \fP(\Omega)$ such that $\E_\Q[\gamma]\leq 0$ holds for all $\gamma \in \Gamma_d$. Then  $S$ is a $\Q$-$\mathbb{F}$-local martingale.
\end{proposition}
\begin{proof}
	First of all, note that 
	 $S=(S^{(1)},\dots,S^{(d)})$ is a $d$-dimensional $\Q$-${\mathbb{F}}$-local martingale if and only if each component  $S^{(i)}$ is a $\Q$-${\mathbb{F}}$-local martingale.  In addition, since each $H \in \mathcal{H}_{s,d,m}({\F}_+)$ is ${\F}_+$-predictable and (locally) bounded, we obtain from \cite[Section~III.6]{JacodShiryaev.03} that the stochastic integral $(H\cdot S)$ is well-defined and satisfies that $(H\cdot S)=\sum_{i=1}^d (H^{(i)}\cdot S^{(i)})$.
	Now, for every $i\in \{1,\dots,d\}$ let
	 \begin{equation*}
	 \mathcal{H}_{s,1,m}^{(i)}(\F_+):=\Big\{H \in \mathcal{H}_{s,1}(\F_+) \colon (H\cdot S^{(i)}) \geq -m \mbox{ pointwise on } \Omega \times [0,T] \Big\}.
	 \end{equation*}
	Then for every $i\in\{1,\dots,d\}$, we see that
	 any $H \in \mathcal{H}^{(i)}_{s,1,m}({\F}_+)$
	 can be extended to an element $H_d:=(H_d^{(1)},\dots,H_d^{(d)})\in \mathcal{H}_{s,d,m}({\F}_+)$ by setting $(H_d^{(1)},\dots,H_d^{(i)},\dots,H_d^{(d)}):=(0,\dots,H, \dots,0)$ which satisfies for every $\gamma \in C_b(\Omega)$ that 
	$\gamma \leq (H_d \cdot S)_T$ if and only if  $\gamma \leq (H \cdot S^{(i)})_T$. Therefore, we conclude that for each $i\in \{1,\dots,d\}$ we have that
	\begin{equation*}
	\Gamma^{(i)}\!:
	=\!\Big\{ \gamma \in C_b(\Omega)\colon \mbox{there exists } m \in \N, \  H \in \mathcal{H}^{(i)}_{s,1,m}(\F_+) \mbox{ such that } \gamma \leq (H \cdot S^{(i)})_T\Big\}\!\subseteq \Gamma_d.
	\end{equation*} 
	As a consequence, it suffices to prove for each $i\in \{1,\dots,d\}$ that if  $\Q \in \fP(\Omega)$ satisfies $\E_{\Q}[\gamma]\leq 0$ for all $\gamma \in \Gamma^{(i)}$, then $S^{(i)}$ is a $\Q$-${\mathbb{F}}$-local martingale.
	
	Therefore, we fix any component $\cS:=S^{(i)}$
	and   assume 
	that $\E_{\Q}[\gamma]\leq 0$ for all $\gamma \in \Gamma^{(i)}$.
	We want to show that $\cS$ is a $\Q$-${\mathbb{F}}$-local martingale
 with localising sequence 
\begin{equation*}
\tau_m:=\inf\{ t\geq 0 : |\cS_t|\geq m\}\wedge T,
\end{equation*}
i.e.~for every $m\in\mathbb{N}$, the stopped process
\begin{equation*}
\cS^{\tau_m}_t:=\cS_{t\wedge \tau_m}
\end{equation*}
is a $\Q$-${\mathbb{F}}$-martingale. We follow the arguments in \cite[Proposition~5.5]{predictionSet} and \cite[Proposition~4.4]{bartl2019duality}. 

Fix $m\in\mathbb{N}$ and write $\tau:=\tau_m$. We first  show that $\cS^{\tau}$ is an ${\mathbb{F}}$-supermartingale. To that end, let $0\leq s< t\leq T$, and define for every $0<\varepsilon \leq 1$
\begin{align*}
\sigma&:=\inf\{ r\geq  s : |\cS_r|\geq m \}\wedge T,\\
\sigma_\varepsilon&:=\inf\{ r\geq s :  \cS_r> m-\varepsilon \text{ or } \cS_r \leq -m+\varepsilon\}\wedge T.
\end{align*} 
Since both $\tau$ and $\sigma$ are  hitting times of a closed set and $\cS$ has continuous sample paths, 
they are $\F$-stopping times, whereas $\sigma_\varepsilon$, $0<\varepsilon<1$, are $\F_+$-stopping times.

Now, fix an arbitrary $\cF_s$-measurable function $h: \Omega \to [0,1]$. 
Notice that $\sigma=\tau$ on $\{\tau\geq s\}$, so that 
$1_{\{\tau\geq s\}}(\cS_t^{\sigma}-\cS_s)=\cS^{\tau}_t-\cS_s^{\tau}$. Moreover,  $\sigma_\varepsilon$ increases to $\sigma$ as $\varepsilon$ tends to $0$, and therefore $\cS_t^{\sigma_\varepsilon}\to \cS_t^{\sigma}$ as $\cS$ has continuous sample paths.  Since additionally $|\cS_t^{\sigma_\varepsilon}-\cS_s|\leq 2m$, we have that
\begin{equation}\label{eq:append-1}
\E_{{\Q}}[ h(\cS^{\tau}_t-\cS_s^{\tau})]
=
\E_{\Q}[ h\,1_{\{\tau\geq s\}}\,(\cS_t^{\sigma}-\cS_s)] 
=\lim_{\varepsilon\to 0}  \E_{{\Q}}[ h\,1_{\{\tau\geq s\}}\,(\cS_t^{\sigma_\varepsilon}-\cS_s)].
\end{equation}
Recall that $g:=h1_{\{\tau\geq s\}}\colon\Omega\to[0,1]$ is 
${\mathcal{F}}_s$-measurable. Then, by \cite[Lemma~5.3]{predictionSet} there exists a sequence of continuous ${\mathcal{F}}_s$-measurable 
functions $g_k\colon\Omega\to[0,1]$
which converge ${\Q}$-almost surely to $g$.  Moreover, as  $\cS\colon \Omega \to C([0,T],\R)$ is continuous, we get from \cite[Lemma~5.4]{predictionSet} that for every $0<\varepsilon<1$, the function 
\begin{equation}\label{eq:lsc}
\Omega \ni \omega \mapsto \cS_{t\wedge \sigma_\varepsilon(\omega)}(\omega) \in \R
\end{equation}
 is lower semicontinuous.
In particular,  for every fixed $k\in \N$
it holds for $H:=g_k 1_{(\!(s,\sigma_\varepsilon\wedge t]\!]}$ that
\begin{equation}\label{eq:lsc-integral}
\Omega \ni\omega \mapsto (H\cdot \cS)_T(\omega) \in \R\,\text{ is lower semicontinuous.}
\end{equation}
Moreover, the fact that $|\cS.-\cS_s|\leq 2m$ 
on $[\![s,\sigma_\varepsilon]\!]$ and $g_k \in [0,1]$ implies that $(H\cdot \cS) \in [-2m,2m]$ and so
\begin{equation*}
H\in\mathcal{H}^{(i)}_{s,1,2m}(\F_+).
\end{equation*}
In addition, observe that \eqref{eq:lsc-integral} 
ensures that there exists a sequence of bounded continuous functions $\gamma_n\colon\Omega\to\R$ such that 
$\gamma_n\leq (H\cdot \cS)_T$ and $\gamma_n$ increases pointwise to $(H\cdot \cS)_T$.
Therefore we have for each $n \in \N$ that $\gamma_n\in  \Gamma^{(i)}$, hence by assumption  we have for every $\varepsilon \in (0,1)$, $k\in \N$ that
\begin{equation*}
\E_{{\Q}}[g_k(\cS_t^{\sigma_\varepsilon}-\cS_s)] 
= \E_{{\Q}}[ (H\cdot \cS)_T] 
=\sup_n \E_{{\Q}}[\gamma_n] 
\leq 0.
\end{equation*}
We hence conclude from  \eqref{eq:append-1} that
\begin{equation*}
\E_{{\Q}}[h(\cS^{\tau}_t-\cS_s^{\tau})]
=
\lim_{\varepsilon\to 0}  \E_{{\Q}}[h\,1_{\{\tau\geq s\}}\,(\cS_t^{\sigma_\varepsilon}-\cS_s)]
=
\lim_{\varepsilon\to 0}  \lim\limits_{k \to \infty}\E_{{\Q}}[g_k\,(\cS_t^{\sigma_\varepsilon}-\cS_s)]\leq 0.
\end{equation*}
This in turn implies ${\Q}$-a.s.\  that
$\E_{{\Q}}[\cS^{\tau}_t|\mathcal{F}_s]\leq \cS^{\tau}_s$, 
hence $\cS^{\tau}$ is indeed a ${\Q}$-${\F}$-supermartingale. 

By similar arguments one can also show that $\cS^{\tau}$ is a 
${\Q}$-${\F}$-submartingale, hence we conclude that indeed $\cS$ is a ${\Q}$-${\F}$-martingale. 
\end{proof}

\bibliography{stochfin}
\bibliographystyle{plain}


\end{document}